\documentclass[10pt, a4paper, twocolumn]{article}
\usepackage[cm]{fullpage}

\usepackage[utf8]{inputenc}
\usepackage[dvipsnames]{xcolor}
\usepackage{bbold}
\usepackage{algorithm, algpseudocode, color, subcaption}
\usepackage{colortbl}
\usepackage{amsthm}
\usepackage{hyperref}
\usepackage{graphicx}
\usepackage{balance}  
\usepackage{xspace}
\usepackage{mathtools}
\usepackage{pdfcomment}
\usepackage{url}
\usepackage{enumitem}
\usepackage{makecell}

\usepackage{etoolbox}

\usepackage{multirow}
\usepackage{stmaryrd}
\usepackage{mathabx}

\usepackage{longtable}

\newtheorem{theorem}{Theorem}[section]
\newtheorem{definition}{Definition}[section]
\newtheorem{corollary}{Corollary}[theorem]
\newtheorem{lemma}{Lemma}[theorem]

\newtheorem{example}{Example}[section]

\newcommand{\algoname}[1]{\textnormal{\textsc{#1}}}

\newcommand{\database}{{\mathcal D}}
\newcommand{\cnt}{{\tt cnt}}
\newcommand{\relation}{R}
\newcommand{\domain}{\Sigma}
\newcommand{\actdomain}{\domain_{act}}

\newcommand{\reprdomain}{\domain_{repr}}
\newcommand{\attr}{A}
\newcommand{\attributes}{\mathbb{A}}

\newcommand{\CQ}{Q}
\newcommand{\LS}{LS}
\newcommand{\GS}{GS}
\newcommand{\kingtuple}{t^*}
\newcommand{\representative}{representative }
\newcommand{\datalogEQ}{{:-}}
\newcommand{\freq}[1]{f_{#1}}
\newcommand{\tsens}{\delta}
\newcommand{\utsens}{\delta^+}
\newcommand{\dtsens}{\delta^-}
\newcommand{\todo}[1]{{\color{red} TODO: #1}}
\newcommand{\note}[1]{{\color{Peach}{\bfseries Note: #1}}}
\newcommand{\sr}[1]{{\color{blue} SR: #1}}
\newcommand{\xh}[1]{{\color{purple} XH: #1}}

\newcommand{\am}[1]{{\color{cyan} AM: #1}}

\newcommand{\eat}[1]{}
\newcommand{\cut}[1]{}

\newcommand{\tdom}{\mathcal{X}}

\newcommand{\superjoin}{\widetilde{\bowtie}}
\newcommand{\superjoinagg}{\widetilde{\mathbf \bowtie}}
\newcommand{\Qpjoin}{Q_{{\tt path}}}
\newcommand{\qpeg}{Q_{{\tt path-4}}}
\newcommand{\Qajoin}{Q_{{\tt acy}}}
\newcommand{\Qgjoin}{Q_{{\tt gen}}}

\newcommand{\Qgfreq}{Q_{{\tt gf}}}

\newcommand{\R}[1]{R_{#1}}

\newcommand{\topjoin}{\hspace{0pt}{\top}\hspace{-1pt}}
\newcommand{\botjoin}{\hspace{0pt}{\bot}\hspace{-1pt}}

\newcommand{\tree}{\mathcal{T}}

\newcommand{\aggregate}[3]{\gamma_{#1,~#2}~#3}
\newcommand{\groupbyagg}[1]{{\mathbf \gamma}_{#1}}

\newcommand{\TSens}{\algoname{TSens}\xspace}
\newcommand{\TSensDP}{\algoname{TSensDP}\xspace}
\newcommand{\Elastic}{\algoname{Elastic}\xspace}
\newcommand{\PrivSQL}{\algoname{PrivSQL}\xspace}
\newcommand{\TruncTSens}{T_{\TSens}}
\newcommand{\Evaluation}{\algoname{Evaluation}\xspace}

\newcommand{\qtri}{q_\triangle}
\newcommand{\qpath}{q_w}
\newcommand{\qcycle}{q_\square}
\newcommand{\qstar}{q_\Yup}

\newcommand{\subtree}{\mathcal{T}}
\newcommand{\comptree}{\tree^c}

\newcommand{\oneFunction}{\mathbb{1}}

\newcommand{\paratitle}[1]{\vspace{2mm} \noindent \textbf{#1}\xspace}

\usepackage{amsmath}
\DeclareMathOperator*{\argmax}{arg\,max}

\newcommand{\squishlist}{
	\begin{list}{$\bullet$}
		{
			\setlength{\itemsep}{0pt}
			\setlength{\parsep}{3pt}
			\setlength{\topsep}{3pt}
			\setlength{\partopsep}{0pt}
			\setlength{\leftmargin}{1.5em}
			\setlength{\labelwidth}{1em}
			\setlength{\labelsep}{0.5em} } }

	\newcommand{\squishend}{
\end{list}  }

\definecolor{brown}{rgb}{0.8,0.1,0.1}

\definecolor{dartmouthgreen}{rgb}{0.05, 0.5, 0.06}
\definecolor{burntumber}{rgb}{0.54, 0.2, 0.14}
\definecolor{BURNTUMBER}{rgb}{0.54, 0.2, 0.14}

\newcommand{\reva}[1]{#1}
\newcommand{\revb}[1]{#1}
\newcommand{\revc}[1]{#1}
\newcommand{\revm}[1]{#1}

\newcommand{\change}[1]{\revm{#1}}

\begin{document}
\title{Computing Local Sensitivities of Counting Queries with Joins}
\graphicspath{ {./figures/} }
\author{
Yuchao Tao$^\dagger$, Xi He$^\star$, Ashwin Machanavajjhala$^\dagger$, Sudeepa Roy$^\dagger$ \\
yctao@cs.duke.edu, xi.he@uwaterloo.ca,  ashwin@cs.duke.edu, sudeepa@cs.duke.edu \\
$^\dagger$ Duke University, $^\star$ University of Waterloo
}

\date{}

\begin{sloppypar}

\maketitle

\begin{abstract}
Local sensitivity of a query Q given a database instance D, i.e. how much the output Q(D) changes when a tuple is added to D or deleted from D, has many applications including query analysis, outlier detection, and in differential privacy. However, it is NP-hard to find local sensitivity of a conjunctive query in terms of  the size of the query, even for the class of acyclic queries. Although the complexity is polynomial when the query size is fixed, the naive algorithms are not efficient for large databases and queries involving multiple joins.  In this paper, we present a novel approach to compute local sensitivity of counting queries involving join operations by tracking and summarizing tuple sensitivities -- the maximum change a tuple can cause in the query result  when it is added or removed. We give algorithms for the sensitivity problem for full acyclic join queries using join trees, that run in polynomial time in both the size of the database and query for an interesting sub-class of queries, which we call `doubly acyclic queries’ that include path queries, and in polynomial time in combined complexity when the maximum degree in the join tree is bounded.  Our algorithms can be extended to certain non-acyclic queries using generalized hypertree decompositions. We evaluate our approach experimentally, and show applications of our algorithms to obtain better results for differential privacy by orders of magnitude.


\end{abstract}


\section{Introduction}\label{sec:intro}

Understanding how adding or removing a tuple to the relations in the database affects the query output is an important task to many applications~\cite{Kanagal:2011:SAE:1989323.1989411,mcsherry2009privacy,kotsogiannis2019privatesql}. For instance, airline companies need to search for a new flight that can meet the requirements of popular trips. Sales companies should identify the critical part in the production to minimize the number of orders affected by this part. Besides these examples for query explanations, applications of the state-of-the-art privacy guarantee -- differential privacy~\cite{Dwork:2014:AFD:2693052.2693053} -- also need to add sufficient amount of noise to hide the change in the query output due to adding or removing a tuple. In particular, given a database instance $\mathcal{D}$, the maximum change to the query output when one of the given tables in the database adds or deletes a tuple is known as the \emph{local sensitivity} of query on $\mathcal{D}$, and the tuple that matches this maximum change is known as \emph{the most sensitive tuple} in the domain of this database.

Computing the local sensitivity of queries \change{on a single relation is trivial, but it is challenging for queries that involve joins of multiple relations.} These queries join several relations (the base relations or transformed relations) into a single table and count the number of tuples in the join output that satisfy certain predicates.  
For instance, to compute the number of possible connecting flights for a multi-city trip requires a join of flights from the given cities.
\change{
Prior work on provenance for queries and deletion  propagation~\cite{AmsterdamerDT11, Buneman:2002:PDA:543613.543633} focus on removing a tuple in the database, but adding new tuples from the full domain is equally important and even harder especially for complex queries over large  domains.}

\eat{
Computing the local sensitivity of queries that involves joins of multiple relations is challenging. These queries join several relations (the base relations or transformed relations) into a single table and count the number of tuples in the join output that satisfy certain predicates.  For instance, to compute the number of possible connecting flights for a multi-city trip requires a join of flights from the given cities. Similarly, given a database of relations (\texttt{Region}, \texttt{Nation}, \texttt{Customer}, \texttt{Orders}, \texttt{Lineitem}, \texttt{Part}), to compute the number of orders from Europe which purchase \texttt{lineitem} that uses a particular \texttt{part} requires the join of these six relations with corresponding predicates on the \texttt{region} name and the \texttt{part} name. 
Prior work~\cite{DBLP:journals/corr/abs-1207-0872,Laud2018SensitivityAO,arapinis16,DBLP:journals/jpc/EbadiS16, kotsogiannis2019privatesql,mcsherry2009privacy} for differential privacy focus on a global version of sensitivity --- the maximum local sensitivity of all possible database instances. These global sensitivity algorithms either are restricted to a special type of joins~\cite{mcsherry2009privacy, DBLP:journals/jpc/EbadiS16} or output a much larger value than the true local sensitivity for general join queries on unconstrained database instances~\cite{DBLP:journals/corr/abs-1207-0872,Laud2018SensitivityAO,arapinis16,kotsogiannis2019privatesql}. Another line of differentially private algorithms~\cite{DBLP:conf/stoc/NissimRS07,DBLP:conf/sigmod/ChenZ13,Zhang:2015:PRG:2723372.2737785,johnson2018towards} depend on the local sensitivity directly. However, these efforts~\cite{DBLP:conf/stoc/NissimRS07,DBLP:conf/sigmod/ChenZ13, Zhang:2015:PRG:2723372.2737785} either offer no efficient and systematic solution for computing the precise bound or derive very loose bound for the local sensitivity of general queries. Furthermore, there are no efficient solutions to derive the most sensitive tuple from the domain, especially when the full domain is larger than the data size.
}

\change{Therefore, we are motivated to study the {\bf local sensitivity problem for counting queries with joins}. In particular, given a conjunctive counting query $Q$ and a database instance $\mathcal{D}$, we would like to find the local sensitivity of $Q$ on $\mathcal{D}$ and find a tuple $t^*$ from the full domain whose sensitivity matches the local sensitivity. We make the following contributions for this local sensitivity problem. 
\squishlist
\item We show that it is NP-hard to find local sensitivity of a conjunctive query in terms of the size of the query, even for the class of acyclic queries. 
    \item We find an efficient algorithm to solve the sensitivity problem and find the most sensitive tuple for {\em path join queries}, in polynomial time in \emph{combined complexity} \cite{Vardi82}, {\em irrespective of the output size}. This is particularly interesting as the well-known algorithms for acyclic and path join queries \cite{Yannakakis:1981} run in polynomial time in both the size of the input and also the output. 
    \item We present an algorithm, \TSens, that efficiently finds the most sensitive tuple for full acyclic conjunctive queries without self-joins using \emph{join trees}, and for a sub-class of general conjunctive queries though extensions using \emph{generalized hypertree decompositions}. 
    \TSens runs in polynomial time in both the size of the database and query for an interesting sub-class of queries, which we call `doubly acyclic queries’ that generalizes path queries, and in polynomial time in combined complexity when the maximum degree in the join tree is bounded. 
\squishend

This paper also shows an {\bf application of our proposed technique \TSens for differential privacy}. An algorithm satisfies differential privacy if its output is insensitive to adding or removing a tuple in any possible input database. This is usually achieved by injecting sufficient amount of noise to the mechanism in order to hide the changes caused by the most sensitive tuples from the domain. Hence, the utility of the mechanisms crucially depends on the upper bound of the local sensitivity. For general SQL counting queries with joins, current methods either offer no efficient or systematic solutions for computing sensitivity~\cite{NodeDP:TCC2013, blocki2013differentially, DBLP:conf/sigmod/ChenZ13} or severly overestimate the sensitivity resulting in poor accuracy ~\cite{johnson2018towards}. Moreover, some queries are highly sensitive to adding or removing a tuple, and approaches that just add noise calibrated to the sensitivity fail to offer any utility. 

In this paper, we combine \TSens with an effective and general purpose technique for DP query answering, called  \textit{truncation}. Here the query is run on a truncated version of the database where tuples resulting in high sensitivity are removed. While this introduces error in the query answer (bias), it decreases the sensitivity and the noise added, and thus, the overall error. While prior work has used truncation~\cite{mcsherry2009privacy, kotsogiannis2019privatesql}, obtaining high accuracy is challenging as it is nontrivial to determine which tuples to truncate. We show \TSens can solve this challenge:}
\squishlist
   \item Our algorithm \TSens is able to compute the sensitivity of each tuple in the domain. This allows us to develop a new truncation-based differentially private mechanism (called \TSensDP) to answer complex SQL queries by truncating a proper set of sensitive tuples.
    \item \TSens provides tight estimates on the local sensitivity (as much as 2.2 million times better than the state of the art techniques for sensitivity estimation \cite{DBLP:journals/corr/JohnsonNS17}). Moreover, \TSensDP answers queries with significantly lower error than \PrivSQL \cite{kotsogiannis2019privatesql}, a state of the art method for answering SQL queries. 
\squishend

\eat{In this paper, we study the local sensitivity problem for counting queries with joins. In particular, given a conjunctive counting query $Q$ and a database instance $\mathcal{D}$, we would like to find the local sensitivity of $Q$ on $\mathcal{D}$ and find a tuple $t^*$ from the full domain whose sensitivity matches the local sensitivity. 


We make the following contributions in the paper: 
\squishlist
\item We show that it is NP-hard to find local sensitivity of a conjunctive query in terms of the size of the query, even for the class of acyclic queries. 
    \item We find an efficient algorithm to solve the sensitivity problem and find the most sensitive tuple for {\em path join queries}, in polynomial time in \emph{combined complexity} \cite{Vardi82}, {\em irrespective of the output size}. This is particularly interesting as the well-known algorithms for acyclic and path join queries \cite{Yannakakis:1981} run in polynomial time in both the size of the input and also the output. 
    \item We present an algorithm, \TSens, that efficiently finds the most sensitive tuple for full acyclic conjunctive queries without self-joins using \emph{join trees}, and for a sub-class of general conjunctive queries though extensions using \emph{generalized hypertree decompositions}. 
    \TSens runs in polynomial time in both the size of the database and query for an interesting sub-class of queries, which we call `doubly acyclic queries’ that generalizes path queries, and in polynomial time in combined complexity when the maximum degree in the join tree is bounded.  
    \item Our algorithm \TSens is able to compute the sensitivity of each tuple in the domain. This allows us to develop a new truncation-based differentially private mechanism (called \TSensDP) to answer complex SQL queries.
    \item \TSens provides tight estimates on the local sensitivity (as much as 2.2 million times better than the state of the art techniques for sensitivity estimation \cite{DBLP:journals/corr/JohnsonNS17}). Moreover, \TSensDP answers queries with significantly lower error than \PrivSQL \cite{kotsogiannis2019privatesql}, a state of the art method for answering SQL queries. 
\squishend
}

\noindent\textbf{Organization:}
We discuss preliminaries and state the problem in Section~\ref{sec:prelim}. We discuss complexity of our problem in Section~\ref{sec:complexity}. Section~\ref{sec:path} and \ref{sec:general} respectively give algorithms for path join and acyclic conjunctive queries with  possible extensions. These algorithms are used to construct a differentially private mechanism in Section~\ref{sec:app_dp}.  Section~\ref{sec:experiments} presents an experimental evaluation of our approach. Related work and future direction are discussed in Sections~\ref{sec:related} and \ref{sec:conclusion}.

\eat{
\xh{***old text***}\\
Differential privacy~\cite{Dwork:2014:AFD:2693052.2693053}
has emerged as the state-of-the-art
techniques that allow querying over sensitive personal data
with provable privacy guarantees.
In the applications of differential privacy,
query sensitivity analysis is a very important task.
Sensitivity measures the changes in the query output
caused by adding or removing of tuples in the query input.
Sensitivity-based differentially private mechanisms
add a noise drawn from a sensitivity-dependent distribution
to the query answer to ensure the desired privacy guarantee.
The larger the sensitivity is, the larger the noise will be
added to the true query answer.


A large class of mechanisms such as Laplace mechanism~\cite{DBLP:conf/tcc/DworkMNS06}
use {\em global sensitivity} of a query,
which measures the maximum change in the query output
when adding or removing a tuple from any possible database.
This measure is independent of the input database.
At worst case database, the change in the query output by adding or removing a tuple
can be extremely high or unbounded and hence the global sensitivity is high or unbounded,
though this change is small for most of real databases.
Hence, another class of mechanisms
such as smooth-sensitivity mechanism~\cite{DBLP:conf/stoc/NissimRS07}
are based on {\em local sensitivity} of a query,
which measures the maximum change in the query output
when adding or removing a tuple from a given database.
As local sensitivity of a query depends on the true database instance,
and mechanisms that based on the local sensitivity
require additional smoothing steps 
on the local sensitivity to achieve differential privacy. 


However, for complex relational database queries, the estimation of local sensitivity can be very costly.
Existing work either considers a much higher upper bound for local sensitivity (called elastic sensitivity~\cite{DBLP:journals/corr/JohnsonNS17})
or applies a lossy transformation on the input data in order to match a pre-computed sensitivity (called down sensitivity~\cite{DBLP:conf/sigmod/ChenZ13}).
The first approach can still add a unnecessarily large noise to the query answer while the second approach introduces bias to the query answer.
In this work, we show the hardness of computing local sensitivity exactly for general Selection-Projection-Join-Aggregate (SPJA) queries and propose efficient algorithms to compute the local sensitivity with tighter bounds for a slightly constrained class of SPJA queries.
Based on efficient local sensitivity estimation algorithms, we design differentially private mechanisms that achieve better utility than prior approaches while remaining efficient.

\todo{The contributions of this paper are: 
\begin{itemize}
    \item hardness ...
    \item PTime algorithm for path queries ...
    \item ptime algorithm (parameterized by treewidth?) for general queries ... 
    \item Applications to ...
\end{itemize}}

Our paper is organized as follows. We first present background and notations in Section~\ref{sec:prelim} and then define our problem in Section~\ref{sec:ps}. In Section~\ref{sec:computesensSPJA}, we show the hardness result for general SPJA queries and efficient algorithms SPJA queries with certain conditions. Then, we present applications of local sensitivity estimation in Section~\ref{sec:app}, in particular, designing differentially mechanisms for answering SPJA queries. Section~\ref{sec:eval} shows our empirical study of our approaches with real world data and query workloads.
Section~\ref{sec:related} illustrates scenarios where prior approaches can fail to provide good utility and Section~\ref{sec:discussion} provides discussion on future work (including databases with constraints).

}

\eat{
Differential privacy aims to protect individual privacy for data analysis.
However, existing work either do not support general database queries in real-world SQL-based analytics systems,
or answer them with poor utility guarantee.
\todo{Write in plain words.}
Consider a database consisting of $k$ tables $(R_1,\ldots,R_k)$
with domain $(\tdom_1^* \times \cdots \times \tdom_k^*)$.
Given a query $q:\times_{j\in Z} \tdom_{j}^*  \rightarrow \mathbb{R}^l$
over a subset of tables in the database, where $Z\subseteq \{1,\ldots,k\}$,
and a set of privacy parameters $\{\epsilon_i |i=1,\ldots,k\}$,
we would like to output a mechanism $M$ from a class of query sensitivity based mechanism $\mathcal{M}^{\Delta q}$
on instance $(x_1,\ldots,x_k)\in (\tdom_1^*\times\cdots\times\tdom_k^*)$
for answering $q$ such that
(i) the privacy loss of $M$ over each table $R_i$ is $\epsilon_i$,
(ii) the output error is minimized, i.e.
\begin{equation}
\min_{M \in \mathcal{M}^{\Delta q}} \|M(x_{Z})-q(x_{Z})\|,
\end{equation}
where $x_{Z} = \{x_j | \forall j\in Z\}$.
Our work consists of two parts: (i) computing the sensitivity of a query;
(ii) constructing a class of query sensitivity based mechanisms.

Two key limitations of the prior work including PINQ, Elastic sensitivity are:
(i) the sensitivity analysis of a given query is operator based,
independent of the order of the operators and the data, and hence can lead to
a very large sensitivity;
(ii) the tables considered in the prior work have no constraints (dependencies between tables).

In this work, we focus on a wide class of
SQL queries over a database of multi-tables (with constraints)
and design efficient sensitivity based differetially private mechanisms
to answer these queries with high utility.

\ \\ {\bf Contributions}.
The first group of algorithms \todo{[YC]What algorithms?} rewrite the queries
by considering the constraints (dependencies between multiple tables)
such that the sensitivity analysis can be the same as the case without constraints.
The second group of algorithms aim to compute the sensitivity of the transformed queries.
We consider the following types of queries:
\begin{itemize}
\item $SPJ$: we aim to (i) provide algorithms to compute sensitivity of acyclic queries; (ii) show hardness of cyclic queries and provide treewidth algorithms.
\item $SPJA^{end}$: consider aggregate at the end of the query
\item $SPJA$: consider aggregate anywhere
\item $SPJU$: consider union
\item $SPJUD$: consider difference
\end{itemize}
Except SPJUD, all the query types are monotone.
Both local sensitivity and global sensitivity (complexity) will be considered.

There are three classes of mechanisms:
(i) global sensitivity based mechanisms, e.g. Laplace mechanism;
(ii) local sensitivity based mechanisms, as smooth sensitivity may not apply here due to unbounded data size,
we need a new algorithm for this;
(iii) hybrid mechanisms, e.g. limit the multiplicity of join key values, which may introduce bias to the answers.
}

\section{Preliminaries}\label{sec:prelim}
We consider a database instance $\database$ with $m$ tables $\relation_1, \ldots, \relation_m$.
Relation $\relation_i$ has attributes $\attributes_i$ where $k_i = |\attributes_i|$, and $n_i$ tuples. Database $\database$ has attributes $\attributes_\database = \cup_{i = 1}^m \attributes_i$, and $k = |\attributes_\database|$ and $n = \sum_{i = 1}^m n_i$ denotes the total number of attributes and tuples respectively in $\database$. 
For any attribute $\attr \subseteq \attributes_\database$, we use $\domain^\attr$ as the domain of $\attr$. For multiple attributes  $\attributes = \{A_1, \ldots, A_\ell\} \subseteq \attributes_\database$, the domain is  $\domain^\attributes = \domain^{A_1} \times \ldots \times \domain^{A_\ell}$.
For a tuple $t \in \relation_i$ and attribute $\attr \in \attributes_i$, $t.\attr$ denotes the value of attribute in $t$, and for $\attributes \subseteq \attributes_i$, $t.\attributes$ denotes a list of values of the attributes in $\attributes$ with an implicit order. 

\paratitle{Full conjunctive queries without self-joins.} We focus on \reva{counting queries that counts} the number of output tuples \revc{(in bag semantics)} for the class of full conjunctive queries (CQ) without self-joins\footnote{Note that CQs can include the \emph{selection} operator by adding predicates of the form $A = a$, which we discuss in Section~\ref{sec:extensions}.}, which is equivalent to the \emph{natural join in the SQL semantics} (equal values of common attributes in different relations) and has been extensively studied in the literature~\cite{Chekuri:1997,Grohe:2001,arapinis16}. 
A CQ $Q$ can be written as a datalog rule as:
\begin{align*}
    \CQ(\attributes_\database) 
    \datalogEQ 
    \relation_1(\attributes_1),
    \relation_2(\attributes_2),
    \ldots,
    \relation_m(\attributes_m).
\end{align*}
\vspace{-1em}

Here all the attributes $\attributes_\database$ appear in the \emph{head} of the datalog rule, and $\attributes_i \cap \attributes_j \neq \emptyset$ captures natural join. We also call attributes as variables and relations as atoms. We interchangeably use the equivalent relational algebra (RA) form:
$$\CQ = \relation_1 \Join  \cdots \Join \relation_m.$$ 
where $\Join$ with no subscripts refer to natural joins.
We denote $\CQ(\database)$ as the query result about $\CQ$ on $\database$.
For example, in Figure~\ref{fig:full_conjunctive},  given 4 relations $(R_1,R_2,R_3,R_4)$ and attributes $(A,B,C,D,E,F)$, where each attribute has a domain of size 2, the natural join of these relations $Q(A,B,C,D,E,F)$ has an output of a single tuple $(a1,b1,c1,d1,e1,f1)$ (Figure~\ref{fig:full_conjunctive}(b)).


\eat{
\am{Do we need to add the ability to select subsets of rows? Can we handle selections in our algorithms? Remember that Elastic sensitivity can not reduce the local sensitivity based on selections on tables. If our algorithms can handle selections, it will be a \textit{major} contribution}
It can also be written in the relation algebra form as derived from the join tree. We interchangeably use the datalog and relation algebra notation.
}

\eat{
A CQ query is called acyclic iff there exists a join tree that all attributes are connected. \note{Shall we define it by GYO(H(q))? H(q) denotes the hypergraph of q.}

\todo{Add an example of a CQ query}
}

\subsection{Problem Statement}\label{sec:problem}
\paratitle{Tuple and Local Sensitivity.} Tuple and local sensitivity of a counting query measure the (maximum) possible change in the number of output tuples when a tuple is added to the database or is removed from the database, and are defined as follows.  For two relations $R, R'$ with the same set of attributes, $R \Delta R' = (R - R') \cup (R' - R)$ is the symmetric set difference. 

\begin{definition}[\bf Tuple Sensitivity]\label{def:tuple_sens}
    \revc{Given a tuple $t$ from the domain of any table as $t \in \domain^{\attributes_1} \cup \domain^{\attributes_2} \ldots \cup \domain^{\attributes_m}$, a query $\CQ$, and a database instance $\database$, }
    \squishlist
        \item 
   {\em upward tuple sensitivity} is:\\
   {\small
    $\utsens(t, \CQ, \database) = |  \CQ(\database \cup \{t\}) ~\Delta~ \CQ(\database) |$
    }
    \item \emph{downward tuple sensitivity} is:\\
       {\small
    $\dtsens(t, \CQ, \database) = | \CQ(\database) ~\Delta~ \CQ(\database \setminus \{t\})| $
    }
    \item \emph{tuple sensitivity} is:\\
       {\small
    $\tsens(t, \CQ, \database) = \max{(\utsens(t, \CQ, \database),~ \dtsens(t, \CQ, \database))}$
    }
   \squishend
    We often drop $t$, $\CQ$, and $\database$ and simply use $\utsens$, $\dtsens$, or $\tsens$.
\end{definition}


\begin{definition}[\bf Local Sensitivity]\label{def:local_sens}
    Given a query $\CQ$ and a database instance $\database$, the local sensitivity is defined as the maximum tuple sensitivity:
    $$\LS(\CQ, \database) = \max_{\revc{t \in \domain^{\attributes_1} \cup \domain^{\attributes_2} \ldots \cup \domain^{\attributes_m} } } \tsens(t, \CQ, \database) $$
\end{definition}


\begin{example}\label{eg:localsens}
\change{In Figure~\ref{fig:full_conjunctive},}
the tuple $(a1, b1, c1)$ in $R_1$ has a downward tuple sensitivity of 1 as removing this tuple will decrease the join output size by 1. Similarly, the tuple $(a2, b2, c1)$ from the full domain of $R_1$ has a downward tuple sensitivity of $0$ as no such tuple exists in the given database instance. On the other hand, the tuple $(a2, b2, c1)$ has an upward tuple sensitivity $4$, as adding this tuple will increase the output size by 4. \change{To compute the local sensitivity of this query on the database instance given in Figure~\ref{fig:full_conjunctive}, we need to find the largest possible change to the output size when adding or removing any possible tuple from the domain.}
The local sensitivity of this query equals to 4, and the most sensitive tuple is $(a2, b2, c1)$ in $R_1$.
\end{example}

\eat{\am{Should we make this a subsection and call it \textit{``problem statement"}. We can state the problem and then revisit the applications we might have mentioned in the intro.}

\paratitle{Sensitivity Problem} Below we define the sensitivity problem as looking for the most sensitive tuple from the full domain for the general join query and it's sensitivity as the local sensitivity.
}

\begin{definition}[{\bf The Local Sensitivity Problem}]
Given a query $\CQ$ and a database instance $\database$, the local sensitivity problem aims to find the local sensitivity $\LS(\CQ, \database)$ of $\CQ$ on $\database$, and also find a tuple $\kingtuple$ whose tuple sensitivity matches the local sensitivity. 
\end{definition}

The problem is trivial when there is only one relation $R$ in the database and $Q(R) = R$, since the local sensitivity is always 1 and any tuple can be the most sensitive tuple. In this paper, we focus on full CQs on multiple relations involving multiple joins.

\subsection{Acyclic Queries}\label{sec:acyclic}
One sub-class of CQs that has been studied in depth in the literature is the class of \emph{acyclic queries} \cite{Beeri+1983, Fagin1983, Abiteboul+1995}, which we consider as one of the classes of queries in this paper. There are different notions of acyclicity \cite{Fagin1983}, however, in this paper we will use one of the standard notions based on \emph{GYO decompositions} (from Graham-Yu-Ozsoyoglu) \cite{Abiteboul+1995}. 
\par
Given a CQ $Q$, the query hypergraph has  all the variables or attributes as vertices, and relations appearing in the body of the query as edges. An {\em ear} is a hyperedge $h$ whose vertices can be divided into two groups that (i) either exclusively belong to $h$, or (ii) are completely contained in  another hyperedge $h'$. The GYO-decomposition algorithm  repeatedly picks an ear from the hypergraph, removes the vertices that are exclusively in the ear, and then  removes the ear from the hypergraph, until the hypergraph is empty or no more ears are found. A CQ is acyclic if the GYO-decomposition algorithm  returns an empty hypergraph. Further, the decomposition algorithm  results in a \emph{join-tree}, which will be described next (assuming the  query hypergraph  is connected, otherwise a join-forest is obtained). 
\par
\paratitle{Join-trees.~} A \emph{join-tree} $T$ for a CQ whose hypergraph is connected satisfies the following property: for any two relations $R_i, R_j$ appearing in the body of the query such that $\attributes_i \cap \attributes_j \neq \emptyset$, all attributes in the intersection appear on a unique path from $R_i$ to  $R_j$ in the tree. A join-tree can be obtained for an acyclic query from a GYO-decomposition by adding an edge from relation $R_i$ to another relation $R_j$, when the hyperedge for $R_i$ is being  eliminated as an ear, and all the vertices that do not exclusively belong to $R_i$ belong to $R_j$.   It is well-known that joins on acyclic queries can be computed in polynomial time in the size of the query and the input (combined complexity, see Section~\ref{sec:complexity}). The output can be generated by two passes on a join-tree using semi-join operators \cite{Abiteboul+1995}. Figure~\ref{fig:gyo} shows the hypergraph of the query $Q(A,B,C,D,E,F)$ in Figure~\ref{fig:full_conjunctive} and its GYO decomposition.


\begin{figure}[t]
	\centering
	\begin{subfigure}[b]{0.25\textwidth}
		\centering
		\includegraphics[width=\textwidth, trim={0 100 0 0},clip]{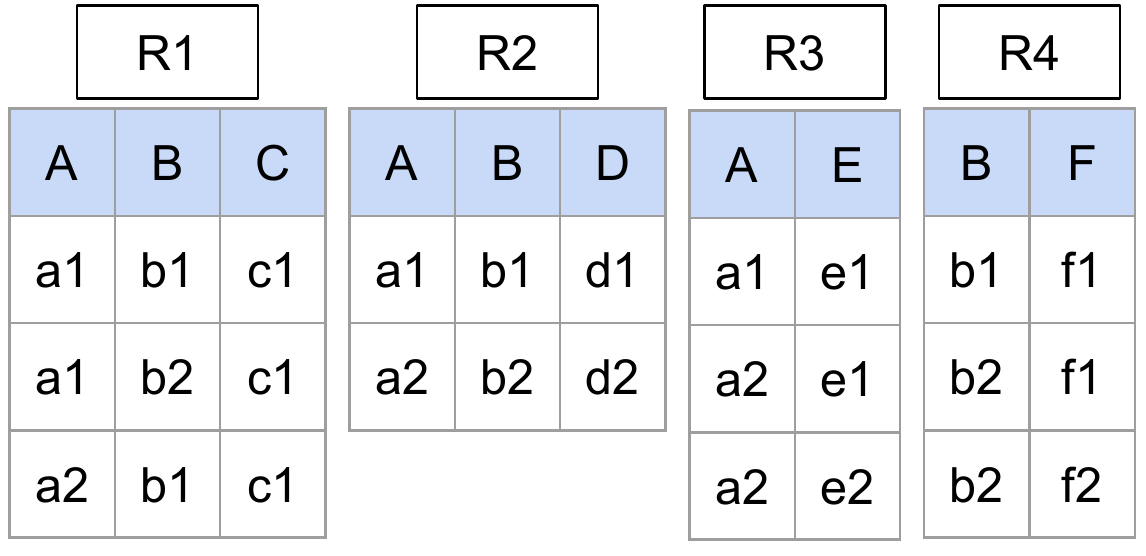}
		\vspace{-15pt}
		\caption{a database instance}
	\end{subfigure}
	\begin{subfigure}[b]{0.2\textwidth}
		\centering
		\includegraphics[width=\textwidth, trim={25 0 25 164},clip]{{"figures/Acyclic_Join_Query_Example"}.pdf}
		\vspace{-15pt}
		\caption{the join result}
	\end{subfigure}
	\vspace{-10pt}
	\caption{An example for a full conjunctive query} 
	\vspace{-12pt}
	\label{fig:full_conjunctive}
\end{figure}

\begin{figure}[t]
	\centering
	\begin{subfigure}[b]{0.2\textwidth}
		\centering
		\includegraphics[width=\textwidth, trim={0 0 200 0},clip]{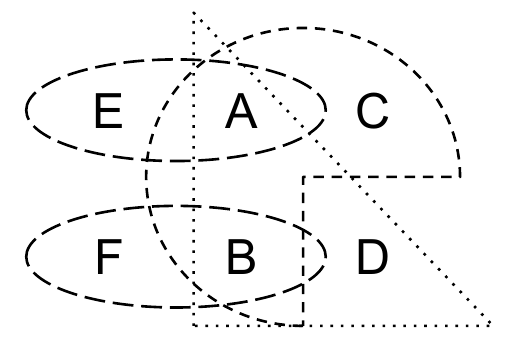}
		\vspace{-15pt}
		\caption{the hypergraph}
	\end{subfigure}
	\begin{subfigure}[b]{0.2\textwidth}
		\centering
		\includegraphics[width=\textwidth, trim={150 0 0 0},clip]{{"figures/GYO_Example"}.pdf}
		\vspace{-15pt}
		\caption{the join tree}
	\end{subfigure}
	\vspace{-10pt}
	\caption{The GYO decomposition of the query in Figure~\ref{fig:full_conjunctive} (left) and the resulted join tree (right). \revc{Here R3(AE), R4(BF) and R4(ABD) are all ears of R1(ABC), so we remove them from the hypergraph and connect them to R1(ABC) in the join tree.} }
	\vspace{-12pt}
	\label{fig:gyo}
\end{figure}

\section{Complexity Analysis}\label{sec:complexity}

\paratitle{Query, Data, and Combined Complexity.~} For evaluation of database queries, both the query size (the number of relations and attributes as $m$ and $k$) and the instance size (the number of tuples $n$) are inputs, and therefore based on the parameter that is considered as variable, three different notions of complexity are considered \cite{Vardi82}.  \emph{Query or expression complexity} treats the query size ($m, k$) as a variable and the data size ($n$) as a  constant. \emph{Data complexity} treats data size as a variable and query size as a constant, whereas \emph{combined complexity}  treats both query 
and data size as variables. It is known that even query evaluation for general CQs is NP-hard for query and combined complexity (e.g., by a simple reduction from clique), but has polynomial data complexity. 

\subsection{Polynomial Data Complexity}\label{sec:poly-data}
The naive solution of computing the local sensitivity is to check the tuple sensitivity of all possible tuples from all tables. \reva{While for downward tuple sensitivity, we need to consider deletions of at most $n$ tuples from the database, for the upward tuple sensitivity when we consider inserting a tuple, the domain of a possible tuple can be arbitrarily larger than $n$ (and even infinite if any attribute has infinite, e.g., integer domain). However, we show below that we can always have a polynomial data complexity by narrowing down the domain of interest}.

\par
\emph{Active domain} of an attribute with respect to a given instance typically refers to the set of values of that attribute appearing in the instance. In our context, given an instance $\database$, a relation $R_i$ in $\database$, and an attribute $\attr \in \attributes_i$, we use $\actdomain^{A, i} = \cup_{t \in R_i} t.A \subseteq \domain^{\attributes_i}$ to denote the active domain of $A$ with respect to $R_i$ in $\database$. If an attribute $A$ appears in two relations $R_i, R_j$, it may happen that $\actdomain^{A, i} \neq \actdomain^{A, j}$.

\eat{
\sr{edited up to here.. PLEASE DO NOT READ OR EDIT THIS SECTION}

\am{There are several concepts here -- full, active, idle, representative, bounded, completementary trivial. Do we need all these terms here? There is no definition of "bounded", and representative and complementary-trivial seem to mean the same thing. 

Shouldn't representative depend on the query (the notation and the discussion after Def 3.1 does not mention a query)? If a query joins two datasets or has a selection then the size of he representative domain is smaller than if the query joins three datasets. Also, definition 3.1 is undefined in the absence of a query. 

I wonder if Lemma 3.1 should just say: ``$\reprdomain^{\attributes_i}$ as defined in Equation?? is $\representative$ and has size $O(n^k)$, where $n$ is the number of rows in each relation and $k$ is the number of attributes in the query (?).''}
}

\eat{
For this $v$, We also define the frequency of $v$ as
$ \freq{v}^i = \sum_{t \in \relation_i} \oneFunction_{[t.\attributes = v]} $.
If the context is clear, we simply write it as $\freq{v}$.
}

\revc{For the upward tuple sensitivity, we only consider tuples that can possibly change the result after the insertion, so its attribute values should appear in all other relations. We define representative domain 
to capture this intuition:}

\begin{definition}[\textbf{Representative domain}]\label{def:repr} 
Given an instance $\database$, a relation $R_i$ in $\database$, and an attribute $\attr \in \attributes_i$, we define the \emph{representative domain} of $\attr$ 
with respect to  $R_i$ as $ \reprdomain^{\attr, i} =  \bigcap_{j ~:~ \attr \in \attributes_j, j \neq i} \actdomain^{A, i}$, if $A$ appears in  at least two relations, and 
    set it as $\{a\}$ for any arbitrary value $a \in \actdomain^{\attr, i}$, if $A$ does not appear in any other relation. 
\par
\revc{
The representative domain for a relation $R_i$, denoted by $\reprdomain^{\attributes_i} \subseteq \domain^{\attributes_i}$, is defined as
$\reprdomain^{\attributes_i} = \reprdomain^{A_1, i} \times \ldots \times \reprdomain^{A_{k_i}, i}$
where $\attributes_i = \{A_1, \cdots, A_{k_i}\}$ are the attributes in $R_i$.
}

\end{definition}

\vspace{-2mm}
\revc{
\begin{example}\label{eg:repr}
 From Figure \ref{fig:full_conjunctive}, the representative domain of $A$ in $R_1$ is $\reprdomain^{\attr, 1} = \actdomain^{\attr, 2} \cap \actdomain^{\attr, 3} = \{a1, a2\} \cap \{a1, a2\} = \{a1, a2\} $
\end{example}
}

We show the following theorem with the proof in appendix.

\begin{theorem}\label{thm:poly-data-complexity}
The local sensitivity of a full CQ $Q$ with respect to a given instance $\database$ can be computed in polynomial time in data complexity.
\end{theorem}

\eat{
\begin{sloppypar}
\begin{proof}
\textbf{(Algorithm)} The algorithm works as follows. First, compute the maximum downward tuple sensitivity $\delta^{-*} = \max_{t \in Q(\database)}\dtsens(t, \CQ, \database)$  (see Definition~\ref{def:tuple_sens}), and note the tuple giving the max value. Next, compute the maximum upward 
tuple sensitivity as $\delta^{+*}$  = $\max_{i \in \{1, \cdots, m\}} \max_{t \in \reprdomain^{\attributes_i}}\delta^{+*}(t, \CQ, \database)$, again noting the tuple giving the maximum value. Return $\delta^{*} = \max(\delta^{+*} , \delta^{-*})$ along with the tuple that led to this highest value.  
\par
\textbf{(Correctness)} We omit the proof that this algorithm correctly computes the local sensitivity due to space constraints.
\par
\textbf{(Polynomial data complexity)} Finally we argue that the algorithm runs in time polynomial in $n = |D|$. Note that the active domain of any single attribute $A \in \attributes_D$ in any relation $R_i$ can be computed in time polynomial in $n$ (in $O(n \log n)$ time if we use sorting to remove duplicates), and $|\actdomain^{\attr, i}| \leq n$. Since each relation $R_i$ has at most $k$ attributes,  $|\reprdomain^{\attributes_i}| \leq n^k$.  Hence the above algorithm iterates over polynomial number of choices for $t$, for each $t$ it evaluates the query $Q(D \cup \{t\})$ or $Q(D - \{t\})$, which can be done in polynomial time in $n$. Hence the total time of the above algorithm is also polynomial in $n$.
\end{proof}
\end{sloppypar}
}

\subsection{Combined Complexity: NP-hardness}
Theorem~\ref{thm:poly-data-complexity} shows that the sensitivity problem has polynomial data complexity, but the algorithm may run in $O(m n^k)$ time, which is inefficient even for a small number of relations and attributes. 
Therefore, in this section, we study the combined complexity for the problem and show that the exponential dependency on the query size is essential  under standard complexity assumptions not only for general CQs, but also for acyclic queries, thereby motivating the study of efficient, practical algorithms for the sensitivity problem discussed in the subsequent sections. 

\eat{
\am{Review how combined complexity is different from data complexity, and why we want to study this? For instance, say that an algorithm that is $O(n^k)$ would not scale even for queries that join 3 relations each of size a million rows.}
It is known that evaluating a CQ query can be NP-hard in combined complexity if this query is non-acyclic. However, although evaluating acyclic queries is of polynomial time in the size of the input and output data \note{If it is polynomial to the output size, it means it can be exponential in the combined complexity. Shall we emphasize this? since our solution to the double acyclic query can find LS in polynomial time in combined complexity}, computing local sensitivty of acyclic queries can be NP-hard.
}

\begin{theorem}\label{thm:combined-nphard} 
Given a CQ $\CQ$ and a database $\database$ as input, the sensitivity problem is NP-hard in combined complexity. In particular, checking whether $LS(\CQ, \database)>0$ is NP-hard for combined complexity \reva{even if $Q$ is acyclic}.
\end{theorem}

\eat{
\am{This theorem statement needs to be rewritten. It reads ... the problem if hard for \textit{every} query Q (which would be data complexity)? Rather the theorem statement should read: The problem of checking whether $LS(Q,D)>0$, where $Q$ and $D$ are input query and database respectively is NP-hard with respect to the combined complexity of the inputs $Q$ and $D$.}
}

\cut{
\vspace{-10pt}
\revc{
\begin{proof}(sketch)
For any SAT instance, we transform each clause into a table with attributes matching the variables in the clause. Each table only keeps the values of attributes that makes the clause as true. We add one more empty table $R_0$ with attributes matching all variables, so together all tables form an acyclic query $\CQ(\database)$ and $|\CQ(\database)| = 0$. The problem that whether there is a solution to the SAT instance now is equivalent to whether adding a tuple in $R_0$ results in $|\CQ(\database)| > 0$ or not. The full proof goes to the full version.
\end{proof}
}
}

\eat{
\begin{proof}
\cut{
The local sensitivity of a CQ query is defined as the maximum change to the query output size if adding to one of the tables or removing a tuple from one of the tables.
The changes to the query output size when removing a tuple is upper bounded by the changes to the query output when adding a tuple. Hence, we will just show that checking if the changes to the query output size is greater than 0 if adding a tuple to one of table is NP-hard.
}

We give a reduction from the 3SAT problem. Consider any instance of 3SAT $\phi$ with $s$ clauses
($C_1,\ldots,C_s$) and $\ell$ variables ($v_1,\ldots,v_\ell$), where each clause is disjunction three literals (a variable or its negation), and the goal is to check if the formula $\phi = C_1 \wedge \cdots \wedge C_s$ is satisfiable. 
We create an instance of the sensitivity problem $LS(Q, \database)$ with $s+1$ relations and $\ell$ attributes in total.
For each clause $C_i$ that involves variables $v_{i_1},v_{i_2},v_{i_3}$, we add a table $R_i$ with three Boolean attributes $A_{i_1},A_{i_2},A_{i_3}$, and insert all possible triples of Boolean values that satisfy the clause $C_i$ into $R_i$ in $\database$. For example, if $C_i$ is $v_{2}\lor \bar{v}_{5} \lor \bar{v}_{7}$, then $R_i(A_2, A_5, A_7)$ contains seven Boolean triples
$(0,0,0), (0,0,1),\ldots,(1,1,1)$ {\em except} $(0,1,1)$.
In addition, we create an empty relation $R_0(A_1,\ldots,A_\ell)$, which does not contain any tuple in $\database$. 
The query is:
\begin{small}
\begin{align*}
    \CQ(A_1 ,\ldots, A_\ell) = R_0(A_1, \cdots, A_\ell) \wedge \underset{i = 1, \ldots, m}{\bigwedge} \R{i}(A_{i_1},A_{i_2},A_{i_3})
\end{align*}
\end{small}
Note that $Q$ is acyclic, as all of $R_1, \cdots, R_s$ correspond to ears (see Section~\ref{sec:acyclic}). Further, the reduction is in polynomial time in the number of variables and clauses in $\phi$. Next we argue that $\phi$ is satisfiable if and only if $LS(Q, \database) > 0$.
\par
\textbf{(only if)} Suppose $\phi$ is satisfiable, and $\mathbf{v} = (v_1 = b_1, \cdots, v_\ell = b_\ell)$ is a satisfying assignment. Then the join of $R_1 \Join \cdots \Join R_s$ is not empty and $\mathbf{v}$ belongs to their join result. However, $Q(\database) = \emptyset$ as $R_0 = \emptyset$ in $\database$. Now, if we add a tuple corresponding to $\mathbf{v}$ to $R_0$, then $Q(\database \cup \{\mathbf{v}\})$ is no longer empty (at least contains $\mathbf{v}$), and therefore $LS(Q, \database) > 0$. 
\par
\textbf{(if)} 
Suppose $LS(Q, \database) > 0$. Hence there exists at least one tuple $t$ such that if it  is added to one of the relations $R_0, R_1, \cdots, R_s$, then $|Q(\database \cup \{t\})| > |Q(\database)|$. Since  $Q(\database) = \emptyset$ as $R_0$ is empty, this tuple must be inserted to $R_0$ to have  a non-empty output. Further, the projection of this tuple to $A_{i_1}, A_{i_2}, A_{i_3}$ for relation $R_i$ must match one of the existing seven tuples of $R_i$ in $\database$ to have a non-empty join result. Therefore, this tuple (Boolean values for $v_1, \cdots, v_\ell$) gives a satisfying solution for $\phi$ by satisfying all the clauses, and makes $\phi$ satisfiable.
\end{proof}
}


\revc{The proof can be found at the appendix and some intuitions of hard acyclic queries are discussed in Section~\ref{sec:general-algo}.} Although Theorem~\ref{thm:combined-nphard} gives a negative result even for acyclic queries, the proof suggests that we may get polynomial-time for special classes of acyclic queries. Indeed, as we show in Sections~\ref{sec:path} and \ref{sec:general}, we can get polynomial combined complexity for the sensitivity problem for \emph{path queries}, and  for an interesting sub-class that we call \emph{doubly acyclic queries}. The algorithm uses join trees and works for other full acyclic CQ. It gives polynomial running time in combined complexity when the max degree in the join tree is bounded and can be extended to certain non-acyclic CQs. 

\section{Path Join Query}\label{sec:path}

In this section, we give an efficient algorithm for a special class of acyclic queries called \emph{path join queries} or \emph{path queries} that run in polynomial time in combined complexity, {\em irrespective of the output size} (note that the polynomial combined complexity for query evaluation of acyclic and path queries is polynomial in input query, input database, and also the output size, which can be exponential in the query size).  
A path join query has the following form:
\begin{align*}
    \Qpjoin(\attributes_\database)
    \datalogEQ 
    \relation_1(A_0, A_1),
    \relation_2(A_1, A_2),
    \ldots,
    \relation_{m}(A_{m-1}, A_{m})
\end{align*}
where $\attributes_\database = \{A_0, A_1, \ldots, A_m\}$ and each relation $\relation_i$ contains two attributes: $A_{i-1}$ and $A_i$. 
Note that the above form can be used when two adjacent relations share more than one attributes, since we can replace multiple attributes with a single attribute using combinations of values for multiple attributes. Path joins can capture natural joins in many scenarios, like joining {\tt Students, Enrollment, Courses, TaughtBy, Instructors, $\cdots$} relations, or, joining {\tt Region, Nation, Customer, Orders, and Lineitem} (e.g., in TPC-H data, see Section~\ref{sec:experiments}). In addition, our algorithm for path join queries will give the basic ideas of our algorithms that can handle general CQs discussed in Section~\ref{sec:general}. 

\begin{figure}[t]
	\begin{center}
		\includegraphics[scale=0.6]{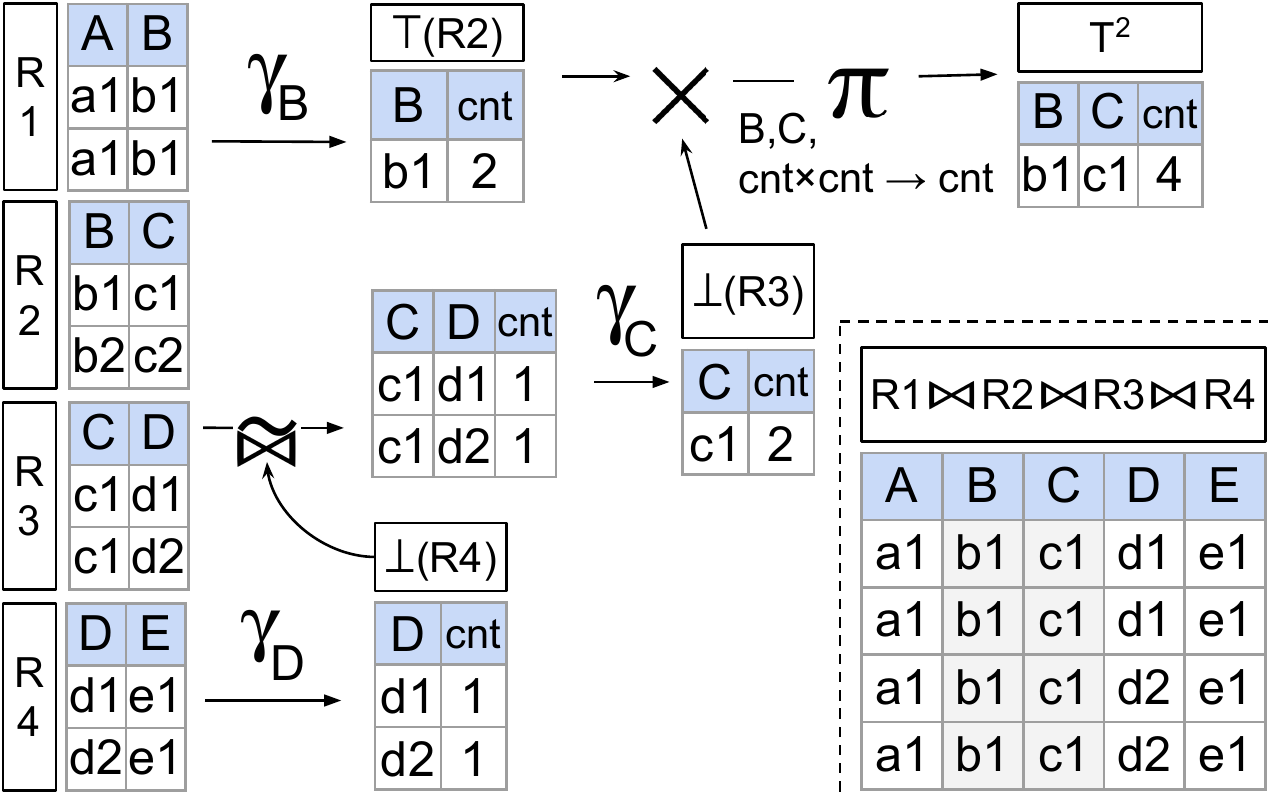}
		\caption{\revc{A path join query $\qpeg(A, B, C, D, E):-R_1(A, B), R_2(B, C), R_3(C, D), R_4(D, E)$ and the procedure of computing tuple sensitivities from $R_2$. 
		}}
		\vspace{-15pt}
		\label{fig:path}
	\end{center}
\end{figure}

\change{
\subsection{Intuition}\label{sec:pj_simple_ineffi}
First, we discuss the basic idea of our algorithm using a toy example of a path query in Figure~\ref{fig:path} with four relations:

\vspace{2pt}
\hspace{-10pt}
\scalebox{0.85}{
$  \qpeg(A, B, C, D, E):-R_1(A, B), R_2(B, C), R_3(C, D), R_4(D, E)$
}
\vspace{-7pt}

The number of output tuples affected by adding or removing a tuple $t$ to any of the relations $R_i$ depends on the number of ways in which $t$ can combine with tuples,  or in this case `join-paths', from the remaining relations. 
\revc{Recall that we are using bag semantics from Section \ref{sec:prelim}, so a `join-path' can repeat multiple times and lead to multiple output tuples.}

\vspace{-2mm}
\begin{example}\label{eg:pathjoin_intuition}
In Figure \ref{fig:path}, if removing the tuple $R_2(b1,c1)$, all the 4 tuples in the current answer of $\qpeg(\database)$ will be removed. These 4 tuples are formed by the join between the 2 tuples from $R_1$ (the ``incoming'' paths ending at $b1$) and the 2 tuples from $R_3\bowtie R_4$ (the ``outgoing'' paths starting at $c1$). On the other hand, if the initial $R_2$ does not have the tuple $(b1,c1)$, inserting $(b1,c1)$ to $R_2$ will add 4 new tuples to the query answer. 
\end{example}
\vspace{-2mm}

It is easy to see that the sensitivity of adding or removing a tuple $(a_i,b_i)\in R_i$ is the product of the number of incoming paths ending in $a_i$ and the number of outgoing paths starting in $b_i$. However, computing sensitivity by enumerating all join paths is inefficient since the number of incoming/outgoing paths can be exponential in the number of relations (and thus not polynomial in combined complexity). We also need to consider tuples from all the relations including tuples that are not in the active domain but can possibly connect a new path, further worsening the runtime. Hence, we propose the following algorithm to avoid repeated query evaluation and capture the effects of adding and removing tuples simultaneously. 
}

\subsection{Efficient Algorithm for Path Queries}\label{sec:pj_effi}

\change{
To efficiently represent the data, we first append each relation with an additional attribute $\cnt$ to record the multiplicity of the other attribute values in  that relation. To keep track of the multiplicity of the incoming paths and outgoing paths for the tuples in $R_i$, we define the following terms.

\paratitle{Topjoin and botjoin.~} We define \emph{topjoin} $\topjoin(\R{i})$ and \emph{botjoin} $\botjoin(\R{i})$ for $\R{i}$ as follows, which respectively compute the multiplicities of the values of  attribute $A_{i-1}$ for the partial path joins from $\R{1}$ to $\R{i-1}$, and $\R{i}$ to $\R{m}$. 
{\small
\begin{align}
\topjoin(\R{i}) &= \groupbyagg{\attr_{i-1}}{\big(\superjoinagg(\relation_1, \ldots ,\relation_{i-1})\big)} \label{eq:pj_topjoin}
    \\
    \botjoin(\R{i}) &= \groupbyagg{\attr_{i-1}}{\big(\superjoinagg(\relation_{i}, \ldots ,\relation_m)\big)} \label{eq:pj_botjoin}
\end{align}}    
The notation $\superjoinagg(R_i, \ldots, R_j)$ for $j>i$ used above is a shorthand of two steps: 
    (a) a natural join among $R_i, R_{i+1},\ldots, R_j$ except the attributes $\cnt$, and 
    (b) a projection of the product of these multiplicity attributes $\cnt$ to a new multiplicity column\footnote{A more systematic way to propagate the multiplicity for arbitrary queries has been discussed in the literature, e.g., \cite{AhmadKKN12, AmsterdamerDT11}.}, i.e., abusing RA expressions:
    $\superjoinagg(R_i,\ldots,R_j)= \pi_{\attributes_i,\ldots,\attributes_j,(R_i.\cnt\times \cdots \times R_j.\cnt)\rightarrow \cnt} (R_i \Join \ldots \Join R_j)$. The group-by operation $\groupbyagg{\attributes}(R)$ computes groups according to a set of attributes $\attributes \subseteq \attributes_R$, and also sums the counts as the new count attribute, i.e., $\groupbyagg{\attributes}(R)= \aggregate{\attributes}{sum(\cnt) \rightarrow \cnt}{(R)}$.
    
\vspace{-2mm}
\begin{example}\label{eg:pathjoin}
In Figure~\ref{fig:path}, the topjoin for $R_2$ is   $\topjoin(R_2)=\groupbyagg{B}{(R_1)}=\{(B:b1,\cnt:2)\}$ and the botjoin for $R_3$ is $\botjoin(R_3)=\groupbyagg{C}{(\superjoinagg(R_3,R_4) )}=\{(C:c1,\cnt:2)\}$. In order to compute the maximum change to the query output by adding/removing a tuple $(b1,c1)$ to/from $R_2$, we can multiply the $\cnt$ of $b1$ from $\topjoin(R_2)$ and the $\cnt$ of $c1$ from  $\botjoin(R_3)$, i.e., $2 * 2 =4$. This is the largest possible change to the query answer if adding or removing a tuple to $R_2$, as the multiplicities of the other values are all smaller than the $\cnt$ of $b1$ and  the $\cnt$ of $c1$.
\end{example}
\vspace{-2mm}
}

\change{
Hence, to compute the most sensitive tuple $t^*_i$ within each $R^i$ just requires the tuple $t^{\top}_i$ from $\topjoin(\R{i})$ with the largest multiplicity and the tuple $t^{\bot}_i$ from $\botjoin(\R{i})$ with the largest multiplicity, i.e.,
{\small
    \begin{equation}\label{eq:topbottuple}
    t^{\top}_i =  \argmax_{t \in \topjoin(\R{i})} t.\cnt  \text{ and }
    t^{\bot}_i =  \argmax_{t \in \botjoin(\R{i+1})} t.\cnt
\end{equation}
}
Then $t^*_i$ takes $(t^{\top}_i.A_{i-1},t^{\bot}_i.A_i)$ and its sensitivity takes $(t^{\top}_i.\cnt * t^{\bot}_i.\cnt)$. For $R_1$, the most sensitive tuple $t^*_1.A_1$ can be derived from the most sensitive tuple in $\botjoin(\R{2})$ and $t^*_1.A_0$ can take any values. Similarly, for $R_m$,  the most sensitive tuple $t^*_m.A_{m-1}$ can be derived from the most sensitive tuple in $\topjoin(\R{m})$ and $t^*_m.A_m$ can take any values.
The most sensitive tuple can be identified from these $m$ tuples $(t_1^*,\ldots,t_m^*)$.

The two relations $\topjoin(\R{i})$ and $ \botjoin(\R{i})$ for deriving $t^*_i$ do not share any attribute, so their join is equivalent to a cross product. Hence, we are not only getting the tuples in the active domain of $R_i$, but also considering all the tuples from its representative domain (Definition~\ref{def:repr}) that can lead to a non-zero local sensitivity by joining with tuples in the other relations, which takes care of both \reva{\textbf{upward}} and \reva{\textbf{downward}} tuple sensitivities. 
}


\paratitle{Algorithm.~} 
\change{Explicitly computing topjoin (\ref{eq:pj_topjoin}) and botjoin (\ref{eq:pj_botjoin}) can require exponential combined complexity, so we give an iterative approach in Algorithm~\ref{algo:ls_pathquery} to compute them in polynomial combined complexity. 
We first compute $\topjoin(\R{2})$ as a base case in the way as topjoin is defined in equation (\ref{eq:pj_topjoin}). Next, we iteratively compute $\topjoin(\R{i})$ for $i =   $ 
3 to m in the algorithm. 
As in the efficient query evaluation for acyclic queries\cite{Abiteboul+1995}, we use \emph{sort-merge joins} to compute the pairwise joins and the then groupby (sort both relations on the join column, join together, then groupby and add the $\cnt$ values).  which can be implemented in 
\revm{$O(n_i \log n_i)$}
time
\revc{as $R_{i-1}$ can join at most one tuple in $\topjoin(R_{i-1})$.} We apply a similar approach to compute botjoin for all relations. In total it takes $O(n \log n)$ time.

After preparing topjoin and botjoin, the third step is to combine them together and find the most sensitive tuple. We first find the tuple $t^{\top}_i$ from $\topjoin(\R{i})$ with the highest count and the tuple $t^{\bot}_i$ from $\botjoin(\R{i})$ with the highest count (Eqn.~\eqref{eq:topbottuple}). Then using these tuples, we can construct the most sensitive tuple $t^*_i$ and its sensitivity for each $T^i$ and identify the most sensitive tuple. \revm{Note that finding the tuple with the highest count in any of these relations can be done in $O(n_i)$ time, taking $O(n)$ time in total. Therefore, the following theorem holds (formal correctness and complexity proofs are deferred to the full version due to lack of space):}


\begin{theorem}\label{thm:path}
 
 Algorithm~\ref{algo:ls_pathquery} solves the sensitivity problem and finds the most sensitive tuple for path join queries. The time complexity is \revm{$O(n \log n)$} where $n$ is the size of the database instance irrespective of the size of the output. 
\end{theorem}

\reva{\textbf{Connection with Yannakakis's algorithm \cite{Yannakakis:1981}:} Algorithm~\ref{algo:ls_pathquery} is inspired by Yannakakis's algorithm \cite{Yannakakis:1981} that computes join results for acyclic queries in (near)-linear time in the size of the input and output, and can be adapted to compute the join size in near-linear $O(n \log n)$  time only in the input size $n$ in a single pass. 
In Algorithm~\ref{algo:ls_pathquery}, we make two passes to compute intermediate topjoins and botjoins, and hence have a similar complexity. We note that, however, this is the \emph{total} time complexity for sensitivity computation considering all possible tuple additions and deletions from all relations. If we naively repeat the $O(n \log n)$ time algorithm inspired by \cite{Yannakakis:1981} to compute the output join size 
for all possible tuple deletions, the time would be multiplied by $n$. Further, if we repeat this algorithm for all possible tuple insertions using the representative domain in Definition~\ref{def:repr}, the time would be (approximately) multiplied by a factor of $O(n^2)$. 
Algorithm~\ref{algo:ls_pathquery} provides an approach to the sensitivity problem using ideas from \cite{Yannakakis:1981} in $O(n \log n)$ time for path queries (we compare these experimentally in Section~\ref{sec:exp_ls}.)
\cut{
Unlike traditional query evaluation which considers active domain, we need to show that our algorithm considers the full representative domain within two passes of the query. } 
}}

However, the above theorem is not generalizable to even all acyclic queries (recall Theorem~\ref{thm:combined-nphard}). In the next section, we give algorithms that can handle acyclic CQs in parameterized polynomial time and run in sub-quadratic time for a class called `doubly acyclic queries' that generalizes path queries.

\begin{algorithm}[t]
\begin{small}
	\caption{\change{Compute Local Sensitivity of a Path Join Query and the corresponding Most Sensitive Tuple}}
	\label{algo:ls_pathquery}
	\begin{algorithmic}[1]
	\Require Path query $\Qpjoin(\attr_0 \ldots \attr_m)$, the database instance $\database$
	\Ensure $\LS(Q, \database)$, and a most sensitive tuple $\kingtuple$
		\Procedure{LSPathJoin}{} 
		\Statex I) Prepare topjoin
		\State $\topjoin(\R{2}) = \groupbyagg{\attr_1}{\relation_1}$ \emph{ /* also adds the $\cnt$ values */} 
		\For{$i=3,\ldots,m$}
		    \State $\topjoin(\R{i}) = 
		    \groupbyagg{A_{i-1}}{\superjoinagg{(\topjoin(\R{i-1}), \relation_{i-1}})}$ \emph{ /* also multiplies and adds the $\cnt$ values */}
		\EndFor
		\Statex II) Prepare botjoin
        \State $\botjoin(\R{m}) = \groupbyagg{\attr_{m-1}}{\relation_m}$ \emph{/* also adds the $\cnt$ values */}
		\For{$i=m-1,\ldots,2$}
		    \State $\botjoin(\R{i}) = 
		    \groupbyagg{A_{i-1}}{\superjoinagg{(\botjoin(\R{i+1}), \relation_i)}}$ 
		    \emph{ /* also multiplies and adds the $\cnt$ values */}
		\EndFor
		\Statex III) Select most sensitive tuple \change{
        \For{$i=1,\ldots,m$}
            \State $t^{\top}_i =  \argmax_{t \in \topjoin(\R{i})} t.\cnt$ 
            \State $t^{\bot}_i =  \argmax_{t \in \botjoin(\R{i+1})} t.\cnt$
            \State $t^*_i = (t^{\top}_i.A_{i-1},t^{\bot}_i.A_i)$ with sensitivity $\cnt=(t^{\top}_i.\cnt * t^{\bot}_i.\cnt)$ \emph{/* when $i=1$ (or $i=m$), $A_{0}$ and $A_{m}$ takes any value and $t^{\top}_1.\cnt=t^{\bot}_m.\cnt=1$. */}
        \EndFor}
        \State $t^* = \argmax_{i=1,\ldots,m} t_i^*.\cnt$
		\State $\LS = t^*.\cnt$
		\State \textbf{return} $\LS, t^*$
		\EndProcedure
	\end{algorithmic}
	\end{small}
\end{algorithm}

\section{Acyclic and Other Join Queries}\label{sec:general}

In this section, we discuss a general solution to  acyclic queries (Section~\ref{sec:general-approach}), and then present an efficient algorithm with additional parameters in the running time complexity of the algorithm (Section~\ref{sec:general-algo}). In Section~\ref{sec:double-acyclic}, we show that a class of queries that we call \emph{doubly-acyclic queries} has a polynomial combined complexity. Last, we discuss several extensions of this algorithm to general cases. 

We consider queries with no self joins; i.e., there are no duplicated relations in the query body. For simplicity, we assume an acyclic full CQ of the following form: 
\begin{align*}
\Qajoin(\attributes_\database) :- \relation_1(\attributes_1), \ldots, \relation_m(\attributes_m).
\end{align*}
We assume that the query hypergraph (Section~\ref{sec:acyclic}) is connected. We also assume that each attribute appears in at least two relations in the body; further, there are no selections in the body, and no projections in the head of the query, i.e., $\attributes_\database = \cup_{i=1}^m \attributes_i$, and the total number of attributes $|\attributes_\database|$ is $k$. These assumptions, except the no-self-join assumption (which introduces new challenges and we leave it as a future direction), are without loss of generality as how they can be relaxed using our algorithm is discussed in Section~\ref{sec:extensions}. 

\subsection{Basic Idea using Join Trees}\label{sec:general-approach}
\revc{Similar to a path join query, the sensitivity of adding or removing a tuple in a relation depends on the number of the incoming/outgoing paths to/from this tuple. To compute the multiplicity of these paths, we represent an acyclic query using a join tree constructed from its GYO decomposition (Section~\ref{sec:acyclic}). For example, given the join tree for a join between 12 relations in Figure~\ref{fig:gj_notations}, in order to compute the sensitivity of tuples in $R_8$ (node 8), we need to construct the join between two groups of relations: (i) the set of relations that are not the descendants of node 8, i.e., $\{11,12,9,10,1,2,7,3,4\}$ --- the incoming paths and (ii) the relations rooted at node 8, i.e., $\{5,6\}$ --- the outgoing paths. }


\change{Formally,} we denote this join tree derived based on GYO decomposition
as $\tree(V,E)$, where the nodes in the tree $V=\{\R{1},\ldots,\R{m}\}$ correspond to relations in the query.  
Let $p(\R{j})$ denote the parent of node $\R{j}$ in the join tree,
$C(\R{j})$ denote the children of $\R{j}$,
and $N(\R{j})$ denote the \emph{neighbors} or \emph{siblings} of $\R{j}$, i.e. $N(\R{j}) = C(p(\R{j})) \setminus \{\R{j}\}$.
We denote $\subtree(\R{j})$ as the relations in the subtree rooted at $\R{j}$, while $\comptree(\R{j})$ is the set of relations in the complement of $\subtree(\R{j})$ on the tree $\tree(V, E)$.

\vspace{-2mm}
\change{
\begin{example}\label{eg:jointree}
In Figure~\ref{fig:gj_notations}, the complementary tree of $R_8$, $\comptree(\R{8})$,  includes $\{11,12,9,10,1,2,7,3,4\}$ and the subtrees at the children of $R_8$ are leaf node 5 and leaf node 6. Computing the joins of these relations can be exponential in the number of the relations or the number of attributes (and thus not have a polynomial combined complexity). We propose an algorithm to make two passes on $\tree$ to efficiently track the incoming/outgoing paths.
\end{example}
}
\vspace{-2mm}

\eat{
To compute the local sensitivity of an acyclic CQ, the basic idea is similar to that in the path join case, where the join and aggregate operators also update the \cnt\ column.  

\squishlist
\item[(i)] For each $i=1,\ldots,m$, compute
{\small
\begin{align}
S^i   =  \superjoin(\{\R{j} \mid \R{j} \neq \R{i}\}) 
 =  \superjoin(\comptree(\R{i}), \{\subtree(\R{j}) \mid \R{j} \in C(\R{i}) \} )\label{eq:general_gj_sub_query}
\end{align}
}
i.e., join the relations that are not in the subtree rooted at $R_i$, with the subtrees rooted at the children of $R_i$.
\item[(ii)] For each $i=1,\ldots, m$, compute 
{\small
\begin{align}\label{eq:general_gj_freq}
T^i = \groupbyagg{\attributes_i}(S^i)
\end{align}
}
\item[(iii)] Finally, return the most sensitive tuple:
{\small
\begin{align} \label{eq:general_gj_most_sensitive_tuple}
t^* = \argmax_{t:t\in T^i, i=1,\ldots,m} t.\cnt
\end{align}
}
\squishend

Note that for path queries, the join tree is a simple path, and therefore \eqref{eq:general_gj_sub_query} reduces to $S^i$ 
as in \eqref{eq:simple-1}.
Similar to the previous section, the size of the sub-queries $S^i$ can be exponential in $m$ or $k$, which is polynomial in data complexity, but may not be efficient for practical purposes. 
To design an efficient algorithm for acyclic CQs, we compress the frequency table by tracking topjoins and botjoins with two passes on $\tree$, which we will discuss next.
}

\subsection{Efficient Algorithm for Acyclic Queries}\label{sec:general-algo}

\eat{
Since $\tree$ is a join tree, for each attribute $A$, relations that contain $A$ always form a connected subtree.  
Hence, all the attributes of  $\R{i}$ that appear in the join tree should be either in the attributes of its complementary tree or the attributes of its descendants. Then applying group by according to the attributes in $\R{i}$ on the join between all the remaining relations gives us the sensitivities of all the tuples in representative domain of $\R{i}$, i.e., 
{\small
\begin{align}\label{eq:general_gj_freq}
T^i &= \groupbyagg{\attributes_i}(\superjoin(\{\R{j} \mid \R{j} \neq \R{i}\})) \nonumber \\ 
&=\groupbyagg{\attributes_i}(\superjoin(\comptree(\R{i}), \{\subtree(\R{j}) \mid \R{j} \in C(\R{i}) \} )) \label{eq:general_gj_freq}
\end{align}
}
The operators $\superjoinagg$ and $\groupbyagg{}$ are the same as the ones used for path join queries in Section~\ref{sec:pj_effi} which take into account multiplicities and allow us to find the most sensitive tuple in $R_i$ $t^*_i$ --- the tuple with the largest multiplicity in $T^i$. At last, we find the most sensitive tuples among all the relations. 
}

{\bf Topjoin and botjoin.~}
\change{To compute the sensitivity of the tuples in $R_i$, we need to evaluate the join between two groups of relations: (i) the complementary tree of $\R{i}$, and (ii) the subtrees rooted at the children of $\R{i}$. These two groups of relations can be represented as  \emph{topjoin} \change{$\topjoin$} and $\emph{botjoin}$ \change{$\botjoin$}:
\begin{align}
\topjoin(\R{i}) &= 
\groupbyagg{\attributes_{i} \cap \attributes_{p(\R{i})}} (\superjoinagg
(\comptree(\R{i})))
\label{eq:gj_topjoin}
\\
\botjoin(\R{i}) &= 
\groupbyagg{\attributes_{i} \cap \attributes_{p(\R{i})}} (\superjoinagg(\tree(\R{i})))
\label{eq:gj_botjoin}
\end{align}
The operators $\superjoinagg$ and $\groupbyagg{}$ are the same as the ones used for path queries in Section~\ref{sec:pj_effi} which take into account multiplicities. 

Since $\tree$ is a join tree, for each attribute $A$, the relations that contain $A$ always form a connected subtree.  Hence, all the attributes of  $\R{i}$ that appear in the join tree should be either in the attributes of its complementary tree or the attributes of its descendants. Then applying group by according to the attributes in $\R{i}$, $\attributes_{\R{i}}$, on the join between all the remaining relations gives
us the sensitivities of all the tuples in representative domain of $\R{i}$, i.e., 
}
\begin{align}\label{eq:gj_fast_node_freq}
T^{i}=  
\groupbyagg{\attributes_{\R{i}}}\big(\superjoin(\topjoin(\R{i}), \{\botjoin(\R{j}) \mid \R{j} \in C(\R{i}) \} )\big)
\end{align}
\change{We name $T^{i}$ the {\emph multiplicity table} of $R_i$. The expression for $T^i$ is simpler if $\R{i}$ is the root or a leaf. We will discuss it in the algorithm below.}




\eat{
Now we can rewrite the Equation~\eqref{eq:general_gj_sub_query} in step (1) as \begin{align}\label{eq:gj_fast_sub_query}
\Qgjoin^i(\attributes_\database) 
=\superjoin(\Qgjoin^{\pi(R_i)}(\attributes_\database) , H(R_i))
\end{align}
Let $v_i$ be the corresponding node of $R_i$ in the hypertree, i.e., $v_j=\pi(R_i)$. This expression joins three types of relations: (i) the complementary tree of $v_j$, (ii) the subtrees rooted at the children of $v_j$, and (iii) the group of cohorts of $R_i$ in $v_j$.
}

\begin{figure}
	\centering
	\includegraphics[scale=0.45]{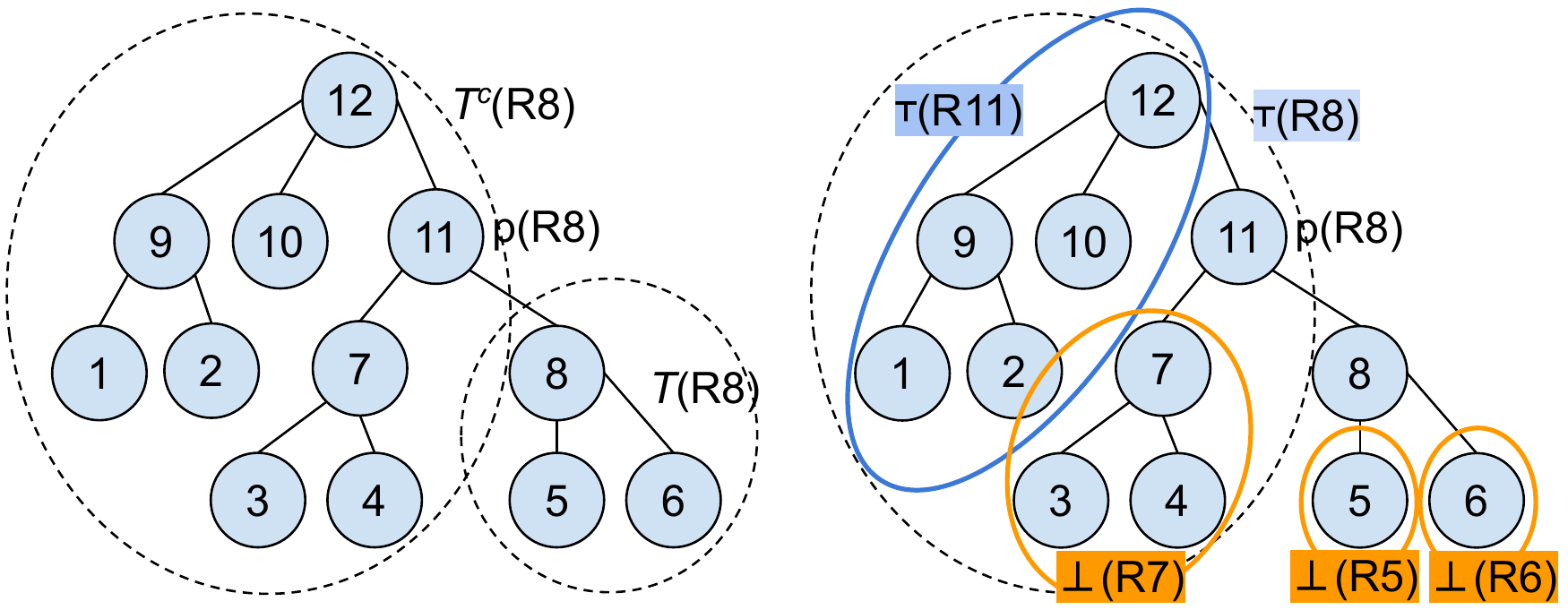}
	\caption{The given join tree consists of 12 relations. The node with number $i$ means $R_i$. \change{The left figure circles the subtree rooted at $R_8$, $\tree(R_8)$, and the complementary subgraph of $\tree(R_8)$, 	$\tree^c(R_8)$. The right figure highlights that the multiplicity table of $R_8$, $T^8$, requires the join between the topjoin of $R_8$, $\topjoin(R_8)$ and the botjoins of all its children, $\{\botjoin(R_5), \botjoin(R_6)\}$.} The topjoin of $R_8$ can be iteratively computed from the join between the topjoin of the parent of $R_8$, i.e., $\topjoin(R_{11})$, and the botjoins of the neighbors of $R_8$, i.e., $\{\botjoin(R_7) \}$.
	}
	\label{fig:gj_notations}
	\vspace{-15pt}
\end{figure}

\eat{
Given a table $R_i$ in $T$, as long as we know the maximum frequency of tuples projected on $\attributes_i$ from the joint output between tables other than $R_i$ in a join tree of $\Qgjoin$, as defined by $\Qgjoin^i$, we immediately get the max change to the final joint output due to adding or removing a tuple to $R_i$, and hence the local sensitivity of the query. We define the following invariants to compute the tuple sensitivity of a general query.
}


\eat{
For each node $v_j$, we denote $\botjoin(v_j)$ and $\topjoin(v_j)$ as tables that computes the frequency of tuples projected on attributes shared by $v_j$ and its parent $p(v_j)$ from $\subtree(v_j)$ and $\comptree(v_j)$.
\begin{align}
\botjoin(v_j) &=  
\groupbyagg{\attributes_{v_j,p(v_j)}} (T(v_j))
\label{eq:gj_botjoin}
\\
\topjoin(v_j) &= 
\groupbyagg{\attributes_{v_j,p(v_j)}} (\comptree(v_j))
\label{eq:gj_topjoin}
\end{align}

Denote the frequency table for $v_j$ as:
\begin{align}\label{eq:gj_fast_node_freq}
\Qgfreq^{v_j}(\attributes_\database) =  
\groupbyagg{\attributes_{v_j}}(\superjoinagg(\topjoin(v_j),
\{\botjoin(v_{j'})~|~v_{j'}\in C(v_j) \}))
\end{align}


Now, equation (\ref{eq:general_gj_freq}) can be rewritten as:
\begin{align}\label{eq:gj_fast_freq}
\Qgfreq^i(\attributes_\database) = \groupbyagg{\attributes_i}(
	\superjoinagg(\groupbyagg{\attributes_i}(\Qgfreq^{\pi(R_i)}), H(\R{i}))
)
\end{align}

\begin{lemma}
	Equation (\ref{eq:general_gj_freq}) and (\ref{eq:gj_fast_freq}) for step (2) are equivalent.
\end{lemma}
\begin{proof}
    
\end{proof}
}


\paratitle{Algorithm.} 
\change{Algorithm~\ref{algo:gj_query} takes as input the join tree $\mathcal{T}$ of the acyclic query and database $\mathcal{D}$ .  It first prepares botjoin and topjoin for each node with an iterative approach.}
To prepare {botjoin} $\bot(\R{i})$ \eqref{eq:gj_botjoin}, we start from the leaf nodes. The botjoin of a leaf node $R_i$ is simply a group by on the common attributes between $R_i$ and its parent node $p(R_i)$ on $R_i$. Next, we compute $\botjoin$ for other nodes in a post-order traversal of the tree with this iterative formula:
\begin{eqnarray}\label{eq:gj_botjoin_iter}
\bot(\R{i}) = \groupbyagg{\attributes_{i} \cap \attributes_{p(\R{i})}} (\superjoinagg(\R{i}, \{\bot(\R{j}) \mid  \R{j} \in C(\R{i}) \}))
\end{eqnarray}
For each $\botjoin(\R{i})$, the join starts with $\R{i}$ and follows by the children of $\R{i}$ one by one.  
\change{
\vspace{-2mm}
\begin{example}
In Figure~\ref{fig:gj_notations}, we first compute the botjoins for all the leaf nodes including $\bot{(\R{1}})$, $\bot{(\R{2}})$, $\bot{(\R{3}})$, $\bot{(\R{4}})$,$\bot{(\R{5}})$, $\bot{(\R{6}})$,  $\bot{(\R{10}})$. Next, if all the children of a node has a computed botjoin, then we can compute the botjoin of this node, e.g. $\botjoin(R_8) = \groupbyagg{\attributes_{R_8} \cap \attributes_{R_{11}}} \left(\superjoin\left(\superjoin(R_8, \bot{(\R{5})}),  \bot{(\R{6})}\right)\right)$, where $R_5$ and $R_6$ are the children of $R_8$, and $R_{11}$ is the parent of $R_8$.
\end{example}
\vspace{-2mm}
}

To prepare {topjoin} $\top(\R{i})$, we start with the children of the root \revb{node}. The topjoin of each child $R_i$ of the root is the join between the root and the botjoins of all its neighbors followed by a group by on the common attributes between $\R{i}$ and the root. Next, we compute topjoin $\topjoin$ for other nodes in a pre-order traversal of the tree with this iterative formula:
\begin{small}
\begin{eqnarray}\label{eq:gj_topjoin_iter}
\hspace{-1pt}
\top\hspace{-1pt}(\R{i})\hspace{-3pt}=
\hspace{-2pt} \groupbyagg{\attributes_{i} \cap \attributes_{p(\R{i})}}
\hspace{-2pt}
(
\superjoinagg(p(\R{i}),
\hspace{-2pt}
\top(p(\R{i})), 
\hspace{-2pt}
\{\bot(\R{j}) \hspace{-2pt} \mid \hspace{-2pt} \R{j}\in N(\R{i})  \}))
\hspace{-5.5pt}
\end{eqnarray}
\end{small}
For each $\top(\R{i})$, the join starts with $p(\R{i})$ and $\top(p(\R{i}))$ and follows by the botjoin of the neighbors of $\R{i}$ one by one. 
For example, computing the topjoin of $R_8$ in  Figure~\ref{fig:gj_notations} requires the join of its parent $R_{11}$, the topjoin of its parent $\topjoin(R_{11})$, and the botjoins of all its neighbors, here $\{\botjoin(R_7)\}$.

\eat{
The algorithm takes as input the join tree $\mathcal{T}$ of the acyclic query and database $\mathcal{D}$ .  It first prepares $\bot(v_j)$ and $\top(v_j)$ for each node. Note that computing  botjoin $\bot(\R{i})$ \eqref{eq:gj_botjoin} and topjoin $\top(\R{i})$ \eqref{eq:gj_topjoin} for each node again may take exponential time in $m$ or $k$, so we give an iterative approach in Algorithm~\ref{algo:gj_query} to compute them efficiently (albeit the time may not be polynomial in combined complexity, recall Theorem~\ref{thm:combined-nphard}). 

After preparing botjoin and topjoin for each node, we use them to compute the frequency tables of each node. 
The final step is to combine the frequency tables to find the most sensitive tuple. 


To prepare {botjoin} $\bot(\R{i})$ \eqref{eq:gj_botjoin}, we start from the leaf nodes. The botjoin of a leaf node $R_i$ is simply a group by on the common attributes between $R_i$ and its parent node $p(R_i)$ on $R_i$. Next, we compute $\botjoin$ for other nodes in a post-order traversal of the tree with this iterative formula:
\begin{eqnarray}\label{eq:gj_botjoin_iter}
\bot(\R{i}) = \groupbyagg{\attributes_{i} \cap \attributes_{p(\R{i})}} (\superjoinagg(\R{i}, \{\bot(\R{j}) \mid  \R{j} \in C(\R{i}) \}))
\end{eqnarray}
For each $\botjoin(\R{i})$, the join starts with $\R{i}$ and follows by the children of $\R{i}$ one by one.  
For example, in Figure~\ref{fig:gj_notations}, the botjoin of $R_7$ requires a join between $R_7$ and the botjoins of its children $R_3$ and $R_4$.

\eat{The correctness follows as replacing $\botjoin(\R{i})$ from the above formula for each $\R{i} \in C(v_i)$ by the definition of topjoin as equation (\ref{eq:gj_topjoin}) shows, and proving by induction.\note{Do we need more steps to show the correctness/equivalence?} \xh{Move this correctness to proof.}}

To prepare {topjoin} $\top(\R{i})$, we start with the children of the root \revb{node}. The topjoin of each child $R_i$ of the root is the join between the root and the botjoins of all its neighbors followed by a group by on the common attributes between $\R{i}$ and the root. Next, we compute topjoin $\topjoin$ for other nodes in a pre-order traversal of the tree with this iterative formula:
\begin{small}
\begin{eqnarray}\label{eq:gj_topjoin_iter}
\top(\R{i}) = \groupbyagg{\attributes_{i} \cap \attributes_{p(\R{i})}}(
\superjoinagg(p(\R{i}), \top(p(\R{i})), \{\bot(\R{j}) \mid \R{j}\in N(\R{i})  \}))
\end{eqnarray}
\end{small}

For each $\top(\R{i})$, the join starts with $p(\R{i})$ and $\top(p(\R{i}))$ and follows by the botjoin of the neighbors of $\R{i}$ one by one. 
For example, computing the topjoin of $R_8$ in  Figure~\ref{fig:gj_notations} requires the join of its parent $R_{11}$, the topjoin of its parent $\topjoin(R_{11})$, and the botjoins of all its neighbors, here $\{\botjoin(R_7)\}$. 

\eat{The correctness follows as replacing $\botjoin(\R{i})$ from the above formula for each $\R{i} \in N(v_i)$ by the definition of botjoin as equation (\ref{eq:gj_botjoin}) shows, and replacing $\topjoin(p(v_i))$ by the definition of topjoin as equation (\ref{eq:gj_topjoin}) shows. Finally, prove it by induction.\note{Do we need more steps to show the correctness/equivalence?} \xh{Move this correctness to proof.}
}
}

After preparing all the topjoins and botjoins, we combine these results to obtain the \change{multiplicity} tables $T^i$  for $i=1,\ldots,m$ based on Eqn.~\eqref{eq:gj_fast_node_freq}. For instance, to compute $T^8$ in Figure~\ref{fig:gj_notations}, we join the topjoin of $R_8$ and the botjoins of all its children, $\{\bot(R_5),\bot(R_6)\}$. \change{This does not require the topjoins of the root node or the botjoins of the leaf nodes.}
We iterate all the \change{multiplicity} tables $T_i$ and find the tuple with maximum $\cnt$. This tuple is returned as the most sensitive tuple with its $\cnt$ as the local sensitivity of this query on the given database instance. 

The  runtime of the algorithm depends on the \emph{max degree} of the tree, which is the maximum children size + 1 (for a non-root node including the parent) of any node in the tree.

\begin{algorithm}[t]
\algrenewcommand\algorithmicindent{1em}%
	\small
	\caption{\change{Compute Local Sensitivity of an acyclic CQ and the corresponding Most Sensitive Tuple}} 
	\label{algo:gj_query}
	\begin{algorithmic}[1]
		\Require Acyclic CQ $\Qajoin(\attributes_\database)$ as a join tree $\tree$, the database $\database$
		\Ensure $\LS(\Qajoin,\database)$, and the most sensitive tuple $\kingtuple$
		\Procedure{LSAcyclicJoin}{} 
		\State I) Compute $\botjoin(\R{i})$ in post-order (leaf to root)
\[
\begin{cases}
   \groupbyagg{\attributes_{i} \cap \attributes_{p(\R{i})}}(\R{i}) ,& \text{if $R_i$ is leaf} \\
   \groupbyagg{\attributes_{i}} \superjoin (\R{i}, \{\botjoin(\R{j}) \mid \R{j} \in C(\R{i})\}), & \text{if $R_i$ is root}\\
   \groupbyagg{\attributes_{i} \cap \attributes_{p(\R{i})}} \superjoin (\R{i}, \{\botjoin(\R{j}) \mid \R{j} \in C(\R{i})\}), &\text{o.w.}
\end{cases}
\]
		\State II) Compute $\topjoin(\R{i})$ in pre-oder (root to leaf) 
\[
\hspace{-10pt}
\begin{cases}
   \emptyset, & \hspace{-27pt} \text{if $R_i$ is root}\\ 
   \groupbyagg{\attributes_{i} \cap \attributes_{p(\R{i})}} \superjoin(p(\R{i}), \{\botjoin(\R{j}) \mid \R{j} \in N(\R{i}) \}), & \hspace{-39pt}\text{if $p(R_i)$ is root}\\
   \groupbyagg{\attributes_{i} \cap \attributes_{p(\R{i})}} \superjoin(p(\R{i}),\topjoin(p(\R{i})),\{\botjoin(\R{j}) \mid \R{j}\in N(\R{i}) \}), & \hspace{3pt} \text{o.w.}
\end{cases}
\]		
		\State III) Prepare multiplicity tables $T^{i}$ of nodes for $i=1, \dots, m$
\[
\begin{cases}
   \groupbyagg{\attributes_{i}} (\topjoin{(\R{i}})), & \text{if $R_i$ is leaf}\\
   \groupbyagg{\attributes_{i}}\superjoin(\{\botjoin(\R{j}) \mid \R{j}\in C(\R{i}) \}), & \text{if $R_i$ is root}\\ 
  \groupbyagg{\attributes_{i}}\superjoin(\topjoin{(\R{i})},\{\botjoin(\R{j}) \mid \R{j} \in C(\R{i}) \}), &\text{o.w.}
\end{cases}
\]			
		\Statex IV) Select  the most sensitive tuple 
		\State $t^* = \argmax_{t:t\in T^i, i=1,\ldots,m} t.\cnt$ and $\LS = t^*.\cnt$ 
		\State \textbf{return} $\LS, t^*$
		\EndProcedure
	\end{algorithmic}
\end{algorithm}

\eat{
\begin{algorithm}[t]
\algrenewcommand\algorithmicindent{1em}%
	\small
	\caption{Compute Local Sensitivity of an acyclic CQ and the corresponding Most Sensitive Tuple} 
	\label{algo:gj_query}
	\begin{algorithmic}[1]
		\Require Acyclic CQ $\Qajoin(\attributes_\database)$ as a join tree $\tree$, 
		the database $\database$
		\Ensure $\LS(\Qajoin,\database)$, and the most sensitive tuple $\kingtuple$
		\Procedure{LSAcyclicJoin}{} 
		\Statex I) Prepare botjoin
		\For{$i$ in postoder}
		\If{$\R{i}$ is a leaf}
		\State $\botjoin(\R{i}) = \groupbyagg{\attributes_{i} \cap \attributes_{p(\R{i})}}(\R{i}) $
		\ElsIf{$\R{i}$ is the root}
		\State $\botjoin(\R{i}) =  \groupbyagg{\attributes_{i}} \superjoin (\R{i}, \{\botjoin(\R{j}) \mid \R{j} \in C(\R{i})\})$
		\Else
		\State $\botjoin(\R{i}) =  \groupbyagg{\attributes_{i} \cap \attributes_{p(\R{i})}} \superjoin (\R{i}, \{\botjoin(\R{j}) \mid \R{j} \in C(\R{i})\})$ 
		\EndIf
		\EndFor
		\Statex II) Prepare topjoin
		\For{$i$ in preoder}
		\If{$p(\R{i})$ is the root}
		\State $\topjoin(\R{i}) = \groupbyagg{\attributes_{i} \cap \attributes_{p(\R{i})}} \superjoin(p(\R{i}), \{\botjoin(\R{j}) \mid \R{j} \in N(\R{i}) \}) $
		\ElsIf{\change{$p(\R{i})$ is not root and $R_i$ is not root} }
		\State $\topjoin(\R{i}) = \groupbyagg{\attributes_{i} \cap \attributes_{p(\R{i})}} \superjoin(p(\R{i}),\topjoin(p(\R{i})),\{\botjoin(\R{j}) \mid \R{j}\in N(\R{i}) \})$
		\EndIf
		\EndFor
		\Statex III) Prepare frequency tables of nodes
		\For{$i=1, \dots, m$}
		\If {$\R{i}$ is the root}
		\State $T^{i} = \groupbyagg{\attributes_{i}}\superjoin(\{\botjoin(\R{j}) \mid \R{j}\in C(\R{i}) \})$
		\ElsIf {$\R{i}$ is a leaf}
		\State $T^{i} = \groupbyagg{\attributes_{i}} (\topjoin{(\R{i}}))$
		\Else 
		\State $T^{i} = \groupbyagg{\attributes_{i}}\superjoin(\topjoin{(\R{i})},\{\botjoin(\R{j}) \mid \R{j} \in C(\R{i}) \})$
		\EndIf
		\EndFor

		\Statex IV) Select  the most sensitive tuple
		\State $t^* = \argmax_{t:t\in T^i, i=1,\ldots,m} t.\cnt$
		\State $\LS = t^*.\cnt$ 
		\State \textbf{return} $\LS, t^*$
		\EndProcedure
	\end{algorithmic}
\end{algorithm}
}


\eat{
	and find the most sensitive tuple, for each relation $\R{i}$, we need to construct $\Qgfreq^i(\attributes_\database)$ according to equation (\ref{eq:gj_fast_node_freq}) and (\ref{eq:gj_fast_freq}). These equations include joins among $\topjoin(\R{i})$, $\botjoin(v_{j'}) ~ \forall v_{j'} \in C(\R{i})$ and $\forall \R{i'} \in H(\R{i})$. Although there exists efficient ways to compute these joins, we postpone this to the future work and leave this as an open problem to analyze the tightest parameterized complexity for these equations. Here we use a naive join algorithm as doing merge-sort join one by one in an arbitrary order. }  

\eat{\note{New rule. We need to avoid cross product. This matches the rule in path join.} For the last step, instead of computing $\Qgfreq^i(\attributes_\database)$ with all joins together, we divide joins into several join groups and get disjoint join results in aspect to attributes. Any two join groups share no common attributes from relations inside the group. We later merge them back by choosing max from each cluster. 
	\note{Double Hypertree Width} Within each join group, we arrange the join order according to the hypertree decomposition of relations within the group. Denote double hypertree width as the max hypertree width within each join group.}

\begin{theorem}\label{thm:gj_query}
	Algorithm~\ref{algo:gj_query} computes the local sensitivity of an acyclic CQ and also finds the corresponding most sensitive tuple. Given $m$ tables with $k$ attributes in total, $n$ tuples in the database instance, and a join tree of the query with max degree  $d$, the time complexity is \revm{$O(m ~d~ n^{d}~\log n)$}.
\end{theorem}

\change{
\begin{proof} 	
    (sketch) If two nodes $R_i$ and $R_j$ share a common attribute $A$, then all the nodes  on the path between $R_i$ and $R_j$ in the join tree also contain $A$.
    Hence, the iterative equations \eqref{eq:gj_botjoin_iter} and \eqref{eq:gj_topjoin_iter} correctly compute the botjoin \eqref{eq:gj_botjoin} and topjoin \eqref{eq:gj_topjoin} by tracking multiplicities through common attributes. 
    
    Now we analyze the running time of the algorithm. Notice that all joins in any topjoin equation \eqref{eq:gj_topjoin_iter} and botjoin equation \eqref{eq:gj_botjoin_iter} have at least one common join attribute,  according to the definition of join tree and the fact that the projection of $\topjoin(R_i)$ and $\botjoin(R_i)$ is always the subset of $\attributes_i$ and $\attributes_{p(R_i)}$. For botjoin \eqref{eq:gj_botjoin_iter}, we join relations with $R_i$ one at a time using sort-merge-join and then do groupby count. The size of each join is always $\leq n_i$ since each tuple $R_i$ can join at most one tuple from any botjoin of its children. In total, it takes $O(d_i n_i \log n_i)$ for each botjoin and $O(m n \log n)$ for all botjoins since the summation of $d_i$ is m and $n_i \leq n$.  
    
    \par
    Next we discuss the running time for step III) in Algorithm~\ref{algo:gj_query}. Unlike the computation for topjoins and botjoins, this step requires joining the botjoins of all the children of a node with the topjoin of that node $R_i$, and all these partial joins may not share any attributes in general (although all the join attributes are still subsets of $\attributes_i$).  Hence, for arbitrary  acyclic joins, there can be at most 
    $d-1$ joins in this step for each $R_i$ where $d$ is the max degree in $\tree$, which can be computed in $O(n^{d} \log n^{d})  = O(d n^d \log n)$ time even by the brute force approach. 
	The total time to compute the multiplicity tables $T^i$ for $m$ relations is $O(m d n^{d} \log n)$. 
	Hence, the total time complexity is $O(m d n^{d} \log n)$. 
\end{proof}
}

\eat{
\begin{proof} 	
    (sketch) If two nodes $R_i$ and $R_j$ share a common attribute $A$, then all the nodes  on the path between $R_i$ and $R_j$ in the join tree also contain $A$.
    Hence, by keeping track of the common attributes with $R_i$, 
    in the topjoin of $R_i$ equations equations \eqref{eq:gj_topjoin} and \eqref{eq:gj_topjoin_iter}, and in the topjoin of the children of $R_i$ 
    \eqref{eq:gj_botjoin} and \eqref{eq:gj_botjoin_iter}
    , it keeps all the sufficient attributes 
    to compute the join in equation \eqref{eq:gj_fast_node_freq}. Hence, $T^i$ can be correctly computed from based on the topjoins and botjoins. 
    
    Now we analyze the running time of the algorithm. First note that, in equations \eqref{eq:gj_botjoin_iter} and \eqref{eq:gj_topjoin_iter}, the maximum size of $\topjoin(R_i)$ or $\botjoin(R_i)$ is $O(n)$. This is because the total number of tuples is $n$, hence any 
    $R_i$ has at most $n$ tuples. These two equations project to $\attributes_{i} \cap \attributes_{p(R_i)} \subseteq \attributes_i$, hence the number of tuples in any 
    $\topjoin(R_i)$ or $\botjoin(R_i)$ cannot exceed $n$. Further, the group by and adding $\cnt$ can be implemented in $O(n \log n)$ time for each $R_i$ by sorting. 
    \par
    Now we discuss how the multiple joins in \eqref{eq:gj_botjoin_iter} and \eqref{eq:gj_topjoin_iter} can be implemented in $O(mn \log n)$ time in total. Consider any $R_i$ and consider \eqref{eq:gj_botjoin_iter}. We join $R_i$ with each 
    $\botjoin(R_j)$, $R_j \in C(R_i)$ one at a time. Note from \eqref{eq:gj_botjoin} that the common attributes in $R_i$ and $\botjoin(R_j)$ is a subset of $\attributes_i$, and in $\botjoin(R_j)$ due to the group-by all the attribute values are unique. Hence we can compute the join using sort-merge-join by sorting these two relations on the join attributes in $O(n \log n)$ time. Similarly in \eqref{eq:gj_topjoin_iter}, we join the relations one  at a time with $p(R_i)$. First, the join attributes in $\topjoin(p(R_i))$ and $p(R_i)$ is a subset of $\attributes_{p(R_i)}$, hence the join again can be computed in $O(n \log n)$ time by sort-merge-join. Similarly, the join attributes in $p(R_i)$ and $R_j$ where $R_j \in N(R_i)$, are a subset of $\attributes_{p(R_i)}$, again taking $O(n \log n)$ for each such join. The total time to compute \eqref{eq:gj_botjoin_iter} and \eqref{eq:gj_topjoin_iter}  is $O(d_i n \log n)$ for each $R_i$ where $d_i$ is the degree of node for $R_i$ in the tree. Since the sum of degree over all nodes in a tree is $O(m)$, where $m = $ the number of relations or the number of nodes in the tree, therefore, the time to compute all the topjoins and botjoins in steps I) and II) in Algorithm~\ref{algo:gj_query} $O(mn \log n)$. 
    \par
    Next we discuss the running time for step III) in Algorithm~\ref{algo:gj_query}. Unlike the computation for topjoins and botjoins, this step requires joining the botjoins of all the children of a node with the topjoin of that node $R_i$, and all these partial joins may not share any attributes in general (although all the join attributes are still subsets of $\attributes_i$).  Hence, for arbitrary  acyclic joins, there can be at most 
    $d-1$ joins in this step for each $R_i$ where $d$ is the max degree in $\tree$, which can be computed in $O(n^{d} \log n^{d})  = O(d n^d \log n)$ time even by the brute force approach. 
	The total time to compute the frequency tables $T^i$ for $m$ relations is $O(m d n^{d} \log n)$. 
	Hence, the total time complexity is $O(m d n^{d} \log n)$. 
\end{proof}
}

\reva{Similar to the discussion in Section~\ref{sec:pj_effi}, the computation of botjoins $\botjoin(R_i)$ and {topjoins} $\top(\R{i})$ in Algorithm~\ref{algo:gj_query} is inspired by Yannakakis's algorithm \cite{Yannakakis:1981}, which can track counts of intermediate tuples from the leaves to the root in a bottom-up pass, whereas in the second top-down pass, we need to traverse the join tree to compute the topjoins $\topjoin({R_i})$. 
As explained earlier for path queries, Algorithm~\ref{algo:gj_query} computes changes in the join size for \emph{all possible} tuple deletions and additions, and naively repeating \cite{Yannakakis:1981} to evaluate query on all possible databases formed by adding or removing a tuple does not give the desired complexity. In fact, \cite{Yannakakis:1981} works in near-linear time in the input size $n$ to output the output join size for any acyclic join query (and has polynomial combined complexity), whereas the sensitivity problem is NP-hard in combined complexity even for acyclic queries as stated in Theorem~\ref{thm:combined-nphard}.}

Given the NP-hardness result in Theorem~\ref{thm:combined-nphard}, we next show an example acyclic query that may take $\omega(mn)$ time
for the $T^i$ step. Suppose we have an acyclic query as $\CQ(A,B,C) :- \R{1}(A,B,C), \R{2}(A,B), \R{3}(B,C), \R{4}(C,A)$ and we want to compute the multiplicity table $T^1$ for $\R{1}$. Given botjoins of $\R{2}(A,B)$, $\R{3}(B,C)$ and $\R{4}(C,A)$, we have a cyclic join among them, and in worst the join size is $O(n^{3/2})$ according to the AGM bound \cite{atserias2008size}. In general, if we replace the children with more complex queries, and if the number of relations (or the degree) is larger, the time to compute this join may be larger. \reva{Note that some of the complexity of this problem comes from the bag semantics considered in our model 
(that is also relevant for applications of sensitivity related to differential privacy), as for set semantics, changes in the join size can be computed more efficiently when a tuple is added or removed from a table. However, for bag semantics, adding any tuple, say to $\R{1}$,  may increase the sensitivity significantly for $\CQ$ (product of the multiplicities of the edges forming the triangle), which adds to the complexity.}

\vspace{-2mm}
\subsection{Doubly Acyclic Join Queries}\label{sec:double-acyclic}


For an acyclic query, if there exists a  join tree $\mathcal{T}$ constructed from the GYO decomposition such that for each node $R_{i}$ in  $\mathcal{T}$, the join between its parent $p(\R{i})$ and its children $C(\R{i})$ is also acyclic, then we say this query is a \emph{doubly acyclic join query}. 
Given this property, the computation of the multiplicity table $T_i$ for $R_i$ involves an acyclic join between the topjoin and botjoins and hence has a time complexity $O(d_i n_i \log n_i)$, where $d_i$ is the node degree of $\R{i}$ in $\mathcal{T}$. Since the sum of all node degrees is $m$ and $n_i \leq n$, the total time complexity to compute $T^i$ for all nodes is $O(m n \log n)$. \revm{When $d_i$ is a constant, such as at most 2, the complexity is written as $O(n \log n)$}, which also matches the total runtime of Algorithm~\ref{algo:gj_query} including the computation of topjoins and botjoins. 


Notice that a path join query is a special case of doubly acyclic join query, because
for each $\R{i}$, 
$\topjoin(\R{i})$ and $\botjoin(\R{i+1})$ (assuming $R_{i+1}$ is the child of $R_i$ in 
$\tree$)
share no attributes.
and therefore is an acyclic join. 
The time complexity of path join queries in  Algorithm~\ref{algo:ls_pathquery} also matches the time complexity of doubly acyclic join queries.

\subsection{Extensions}\label{sec:extensions}
\change{
In this section we briefly discuss how to extend our framework relaxing the assumptions listed at the beginning of Section \ref{sec:general}, and we defer the details in a full version.

\textit{Selections:} We can easily extend Algorithm~\ref{algo:gj_query} to handle queries 
with arbitrary selection conditions (that can be applied to each tuple individually in any relation) in the body of the query by assigning 0 sensitivity to the tuples fail the selection condition. 

\textit{Disconnected join trees:} If the hypergraph of a query is not connected, Algorithm~\ref{algo:gj_query} can be applied to each join tree and merge them back to update each tuple sensitivity. 

\textit{General joins}: For a non-acyclic join query, if there exists a \emph{generalized hypertree decompositions}\cite{adler2007hypertree} such that we can find a join plan from this tree by assigning each relation to one node, Algorithm \ref{algo:gj_query} can be extended by computing the multiplicity table as including other relations within the same node, and the time complexity is parameterized by the max number $p$ of relations within a node as $O(m pd n^{pd} \log n)$. This is implemented in our experiments; q3 from TPC-H queries, and $\qtri$, $\qcycle$ from Facebook queries are all non-acyclic queries and their generalized hypertree decompositions are shown in Figure \ref{fig:queries}.


\textit{Efficient approximations:} We can extend our algorithm to tradeoff accuracy in the sensitivity for better runtime. As our experiment will show, the multiplicity tables that topjoins and botjoins compute can grow quadratically or faster in the input size depending on the query. To make the computation scalable, we can maintain the top $k$ frequent values instead of all the frequencies in the top and botjoins. We can set the frequencies of the rest of the active values in the top and botjoins to the kth largest frequency. This approach gives an upper bound of  tuple sensitivity but can speed up runtime.

\textit{Self Joins:} Acyclic join queries with self-joins can not be captured by our algorithms, because we only allow a relation to appear once in the query. For each relation, we compute the joins for the rest of relations to summarize how tuples from this relation can affect the full join. A possible workaround is to join the repeated base relations as a single and combined relation, run our algorithm, and then link the effect of adding or removing a tuple from the base relation to the combined relation and the effect of adding or removing a tuple from the combined relation to the rest. However, it is challenging to find all possible insertions to the base relation that allows the combined relation to join all possible pairs of "incoming" and "outcoming" path. We defer this line of research to future work. 

\revc{
\textit{Other:} For attributes that appear only once in the query body, we ignore them in Algorithm~\ref{algo:gj_query} but in the end we extrapolate a value for these attributes. 
}
}

\eat{
\hspace{-10pt}\textbf{Extensions to improve query expressibility:}
\vspace{-5pt}
\squishlist
\item \textbf{Non-acyclic join queries} If there exists a \emph{generalized hypertree decompositions}\cite{adler2007hypertree} for a non-acyclic join query such that we can find a join plan from this tree by assigning each relation to one node, Algorithm \ref{algo:gj_query} can be extended by computing the frequency table as including other relations within the same node, and the time complexity is parameterized by the max number $p$ of relations within a node as $O(m pd n^{pd} \log n)$. 

\item  Algorithm~\ref{algo:gj_query} can also handle queries  with arbitrary selection conditions (that can be applied to each tuple individually in any relation) in the body of the query with an extra rule that tuples which cannot satisfy the selection condition has tuple sensitivity as 0.

\item If the hypergraph of query is not connected, Algorithm~\ref{algo:gj_query} can be applied to each join tree and merge them back to update each tuple sensitivity. For attributes that appear only once in the query body, we ignore them in Algorithm~\ref{algo:gj_query} but in the end we extrapolate a value for these attributes. 

\item  Self-joins. Acyclic join queries with self-joins can not be captured by our algorithms, because we only allow one relation appear once in the query. For each relation, we compute the joins for the rest of relations to summarize how tuples from this relation can affect the full join. A possible workaround is to join the repeated base relations as a single and combined relation, run our algorithm, and then link the effect of adding or removing a tuple from the base relation to the combined relation and the effect of adding a removing a tuple from the combined relation to the rest. However, it is not clear what is the representative domain in this case, which makes the max upward tuple sensitivity unknown.

In conclusion, with extensions Algorithm~\ref{algo:gj_query} can handle a more general class of full conjunctive queries (including selections) with a parameterized polynomial time complexity.
\squishend

\hspace{-10pt}\textbf{Extensions to improve runtime.}
\vspace{-5pt}
\squishlist
   
    \item Top k frequency tables. As the experiment shows, the frequency tables that topjoins and botjoins compute can grow quadratically to the input size or even more depending on the queries. To make it scalable, instead of keeping the full frequency table, we can only keep the top k entries and assume that the rest in the active domain all have frequency 1 and others have frequency 0. This gives an upper bound of the tuple sensitivity, which provides a knob on scalability and tightness.
    \item Infer tuple sensitivities through database constraints. Constraints like functional dependencies can imply tuple sensitivities without computing the frequency table. For example, for query Q(A,B,C) :- R1(\underline{A}), R2(\underline{B}, A), R3(\underline{C}, B) where underlined attributes are keys, the tuple sensitivities for tuples from R3 is always at most 1, since it can at most join one tuple in R2, and this tuple in R2 can at most join one tuple in R1. 
\squishend
}

	

\section{Use in Differential Privacy}\label{sec:app_dp}

In this section, we will show how to use our algorithm \TSens (section \ref{sec:general-algo}) for computing sensitivity to develop accurate differentially private algorithms. Section~\ref{subsec:dp} gives a brief overview of differential privacy (DP), and Section~\ref{subsec:dpalgo} discusses how the tuple sensitivity measures can be used to develop accurate DP algorithms.  



\subsection{Differential Privacy}\label{subsec:dp}

Differential privacy (DP)~\cite{Dwork:2006:DP:2097282.2097284} is considered the gold standard for private data analysis. An algorithm satisfies DP if its output is insensitive to adding or removing a tuple in the input database.  Formally, 

\begin{definition}[Differential Privacy]
A mechanism $\mathcal{M}: I \rightarrow \Omega$ is $\epsilon$-differentially private if for any two neighbouring relational database instances $\database, \database'\in I$ and $\forall O \subseteq \Omega$:
$$ \left|\ln (\Pr[\mathcal{M}(\database) \in O]/\Pr[\mathcal{M}(\database') \in O]) \right| \leq \epsilon $$
\end{definition}

When $\database$ is a single relation, all neighboring relations are of the form $\database' = \database - \{t\}$. When $\database$ is a multi-relational database with foreign key constraints, then a neighboring instance $\database'$ is gotten by deleting one tuple $t$ in $\database$'s \textit{primary private relation} and cascadingly deleting other tuples that depend on $t$ through foreign keys \cite{kotsogiannis2019privatesql}.  

\revm{
Differential privacy has been successfully used to publish summary statistics, synthetic data, machine learning models, and answer SQL queries 
\cite{qardaji2014priview, xiao2010differential, abadi2016deep, vaidya2013differentially, mcsherry2009privacy, johnson2018towards,kotsogiannis2019privatesql} . 
It has also been adopted at government \cite{abowd2018us} and commercial organizations 
\cite{erlingsson2014rappor, Wilson:2019GoogleDP, bittau2017prochlo, apple17, Barak}.
}


\paratitle{Laplace Mechanism} is a fundamental building block of DP algorithms \cite{Dwork:2014:AFD:2693052.2693053}. It answers a query $\CQ$ by adding noise drawn from a Laplace distribution scaled to the ratio of the \textit{global sensitivity} of $\CQ$ and the privacy loss parameter $\epsilon$. 


\begin{definition}[Global Sensitivity]
    Given a counting query $\CQ : I \rightarrow \mathcal{R}$, the global sensitivity $GS$ is defined as the max difference of query result from any two neighbouring relational database instances $\database, \database'\in I$ :
    $$\GS(\CQ) = \max_{(\database'-\database)\cup(\database-\database') =\{t\}} \left| \CQ(\database) - \CQ(\database') \right| $$
\end{definition}

\change{Unlike the local sensitivity of a query which depends on the given database instance, the global sensitivity of a query finds the largest possible local sensitivity among all possible database instances.} \revb{Consider the join query in Figure~\ref{fig:full_conjunctive}(b), Example~\ref{eg:localsens} shows that it has a local sensitivity of 4 on the database instance shown in Figure~\ref{fig:full_conjunctive}(a). However, there exist other database instances with a much larger local sensitivity than 4. For example, if $R_2$ in Figure~\ref{fig:full_conjunctive}(a) has 1000 copies of $(a_1,b_1,d_1)$ which results in 1000 copies of $(a_1,b_1,c_1,d_1,e_1,f_1)$ in the join output, removing $(a_1,b_1,c_1)$ from $R_1$ can result in a change of 1000 in the output size.
If there is no \textit{a priori} bound on the number of tuples that share the same join key, the global sensitivity of the query will be unbounded.}

\begin{definition}[Laplace Mechanism]
    Given a counting query $\CQ : I \rightarrow \mathcal{R}$, a database instance $\database \in I$ and a privacy parameter $\epsilon$, the following noisy query result satisfies $\epsilon$-DP:
    \change{$\CQ(\database) + \eta$, where $\eta\sim exp(-\frac{|\eta|\cdot \epsilon}{GS(Q)})$.  }
\end{definition}

\change{The noise $\eta$ has a mean 0 and a variance of $2GS(Q)^2/\epsilon^2$ which increases with the global sensitivity of the query. Hence, this mechanism cannot be directly applied to query with unbounded global sensitivity. Prior work for general join queries either have high performance cost~\cite{NodeDP:TCC2013, blocki2013differentially, DBLP:conf/sigmod/ChenZ13} or suffer from poor accuracy~\cite{johnson2018towards}. One effective and general purpose technique from prior work is \textit{truncation} that executes the query $\CQ$ on a truncated version of the database $T(\database)$ \cite{mcsherry2009privacy, kotsogiannis2019privatesql}. The truncation is done in such a way that $\CQ(T(\cdot))$ has a bounded global sensitivity. For a join, this might mean removing rows from the database such every join key has a bounded selectivity. We will next show a truncation based algorithm for answering SQL aggregation queries with joins based on the tuple sensitivities.
}

\eat{
Note that the sensitivity computed by \TSens is \textit{not} the global sensitivity, but rather the \textit{local} sensitivity on the input database $\database$. 

\paratitle{Aggregate queries expressed as SQL with Joins have \textit{unbounded} global sensitivity}. This is because there is typically no \textit{a priori} bound on the number of tuples that share the same join key. Hence, the Laplace mechanism can not be directly applied. 
Prior work takes two approaches to answering queries with joins under DP -- smooth sensitivity \cite{johnson2018towards} and Lipshitz extensions \cite{NodeDP:TCC2013, blocki2013differentially, DBLP:conf/sigmod/ChenZ13, kotsogiannis2019privatesql}. The former approach tunes noise to a smooth upper bound on the local sensitivity. However, the smooth sensitivity approach was shown to be too inaccurate \cite{kotsogiannis2019privatesql}. 

The Lipshitz extension approach is to rewrite the original query $\CQ$ that has an unbounded sensitivity into a different query $\CQ'$ that (a) has a bounded global sensitivity, and (b) evaluates to roughly the same answer. These techniques are superior to smooth sensitivity based techniques as they introduce some bias (as $\CQ(\database) \neq \CQ'(\database)$) to reduce the noise added to the answer. Many Lipshitz extension techniques are query specific \cite{NodeDP:TCC2013, blocki2013differentially} or fairly inefficient in practice \cite{DBLP:conf/sigmod/ChenZ13}. One class of efficient and general purpose Lipshitz extension techniques is \textit{truncation} that execute the query $\CQ$ on a truncated version of the database $T(\database)$ \cite{mcsherry2009privacy, kotsogiannis2019privatesql}. The truncation is done in such a way that $\CQ(T(\cdot))$ has bounded global sensitivity. For a join, this might mean removing rows from the database such every join key has a bounded selectivity. 
 We will next show a truncation based algorithm for answering SQL aggregation queries with joins based on the tuple sensitivities computed by \TSens. 
}

\subsection{Truncation mechanism with \TSens}\label{subsec:dpalgo}



The idea behind our algorithm is to (a) identify tuples in the database (i.e., in the primary private relation) that have a sensitivity greater than a \textit{sensitivity threshold}, and (b) remove all tuples with sensitivity greater than the sensitivity threshold. 

\begin{definition}[\TSens Truncation]
	Given a query $\CQ$, a database $\database$ with primary private relations $\mathcal{P\relation}$, and a sensitivity threshold $i$, the truncation operator $\TruncTSens$ transforms the database as:
	$$\TruncTSens(\CQ, \database, i) = \left\{ t \in \database \mid t \in \mathcal{P\relation} \Rightarrow \tsens(t, \CQ, \database) \leq i \right\} $$
\end{definition}

The global sensitivity of $\CQ(\TruncTSens(\CQ, \cdot, \tau))$ is $\tau$. If we add or remove a tuple with sensitivity more than $\tau$, the query result does not change as the new tuple will be truncated or has already been truncated. Since the largest possible tuple sensitivity is $\tau$ for any database, the global sensitivity is $\tau$. \change{Hence, given a join query $Q$ with high global sensitivity, we can first apply $\CQ(\TruncTSens(\CQ, \cdot, \tau))$ to the database and then apply Laplace mechanism with smaller noise (due to smaller global sensitivity) on the transformed database. However, the transformed database also introduces bias if too many tuples are truncated. Hence, we would like to find a truncation threshold that minimizes the expected sum of bias and noise.
}

\textbf{Finding truncation threshold.~} \change{ If setting $\tau$ to be the local sensitivity of the query $Q$, then $\CQ(\database) = \CQ(\TruncTSens(\CQ, \cdot, \tau))$, i.e., no bias is introduced. However, using local sensitivity directly violates DP. }
Moreover, the global sensitivity of querying the local sensitivity of a join query is unbounded, we cannot use Laplace mechanism to release a noisy local sensitivity. Instead, line in \PrivSQL \cite{kotsogiannis2019privatesql}, we apply the \textit{sparse vector technique} (SVT) \cite{SVT:VLDB2017} to find the optimal truncation threshold \change{that is close to the local sensitivity.}

\eat{
\begin{algorithm}[t]
	\small
	\caption{Learning truncation threshold ($\CQ$, $\database$, $l$, $\epsilon_{tsens}$) 
	} 
	\label{algo:trunc_threshold}
	\begin{algorithmic}[1]
		\Statex $\epsilon_1 \gets \epsilon_{tsens} / 10$
		\Statex $\epsilon_2 \gets \epsilon_{tsens} - \epsilon_1$
		\Statex $\hat{Q} \gets Q(\database) + Lap(l/\epsilon_1)$
		\Statex $q \gets \{\frac{\CQ(\TruncTSens(\CQ, \database, i)) - \hat{\CQ}}{i} \mid i = 1, \ldots, l-1 \}$  
		\Statex Set $\tau \gets SVT(q, 0, \epsilon_2)$
	\end{algorithmic}
	\label{algo:trunc}
\end{algorithm}
}

For a query $\CQ$ and a database $\database$, let $\ell$ be an upper bound on the local sensitivity. We first release a noisy version of $\CQ(\TruncTSens(\CQ, \database, l))$ as $\hat{\CQ}$ using the Laplace mechanism with global sensitivity as $\ell$. Next, we run the SVT method that checks whether $q_i > 0$ for $i = 1, \ldots, \ell-1$, where 
{\small 
\begin{align*}
q_i = \frac{\CQ(\TruncTSens(\CQ, \database, i)) - \hat{\CQ}}{i} 
\end{align*}
}
Since the global sensitivity of $\CQ(\TruncTSens(\CQ, \database, i))$ is $i$, the global sensitivity of each $q_i$ is a constant 1.  SVT stops the first time (noisy) $q_i$ is above the (noisy) threshold 0 and reports $i$. We take this $i$ as the truncation threshold $\tau$, and answer the query $\CQ$ using $\CQ(\TruncTSens(\CQ, \database, i))$. 
\change{A part of the privacy budget $\epsilon_{tsens}$ is used to release $\hat{Q}$ and run SVT for finding the truncation threshold $\tau$. The rest $\epsilon - \epsilon_{tsens}$ is used to answer the query.}

\begin{theorem}
	 The algorithm that finds the truncation threshold satisfies $\epsilon_{tsens}$-DP and releasing a noisy answer as $\CQ(\TruncTSens(\CQ, \database, \tau)) + Lap(\frac{\tau}{\epsilon - \epsilon_{tsens}})$ satisfies $(\epsilon-\epsilon_{tsens})$-DP. Together the mechanism  satisfies $\epsilon$-DP.
\end{theorem}
\vspace{-1em}

\paratitle{Discussion.~}  
\change{
Our solution is inspired by Wilson et al \cite{Wilson:2019GoogleDP}, but they handle  can only handle a single join (and not self joins), while we can handle a wider sub-class of full conjunctive queries without self joins. Moreover, Wilson et al set the sensitivity threshold manually, while we automatically identify the threshold given an estimated upper bound.  

Our algorithm truncates primary private tables while in \PrivSQL \cite{kotsogiannis2019privatesql}, truncation happens at non-primary private tables. \PrivSQL truncates tuples with high frequencies, but it doesn't mean that they join with the tuple of the highest tuple sensitivity. In contrary, truncation by tuple sensitivity is a finer truncation strategy which reduces global sensitivity and bias at the same time. 

Our algorithm for finding the most sensitive tuples can be easily extended for $\TSens$ by storing the multiplicity table for the primary private table. Our truncation algorithm takes in estimated upper bound of tuple sensitivity $\ell$. Our algorithm will still ensure DP regardless of the value for $\ell$, but the value of $\ell$ can affect the accuracy. We illustrate the impact of $\ell$ on the accuracy in the evaluation.
}

\eat{
\squishlist
\item Our solution is inspired by Wilson et al  \cite{Wilson:2019GoogleDP}, but significantly generalizes it. Wilson et al can only handle a single join (and not self joins), while we can handle a wider sub-class of full conjunctive queries without self joins. 
Moreover, Wilson et al set the sensitivity threshold manually, while we automatically identify the threshold given an estimated upper bound. 

\item In \PrivSQL \cite{kotsogiannis2019privatesql}, truncation happens at non-primary private tables while in this paper it happens at primary private tables. \PrivSQL truncates tuples with high frequencies, but it doesn't mean that they join with the tuple of the highest tuple sensitivity. In contrary, truncation by tuple sensitivity is a finer truncation strategy which reduces global sensitivity and bias at the same time.

\item Truncation by tuple sensitivity requires knowing each tuple sensitivity in the primary private tables. \TSens not only finds the most sensitive tuple, but also summarize the distribution of tuple sensitivity for the full domain. The algorithm preserves several frequency tables for each relation, so each tuple sensitivity of tuples in the live table can be computed by a minor extension of the algorithm.

\item Our truncation algorithm assumes an estimated upper bound of tuple sensitivity $\ell$ is given. Whether $\ell$ is larger or smaller than the true local sensitivity, our algorithm will still ensure differential privacy. If $\ell$ is much larger than the true local sensitivity, then we will estimate the bias incorrectly. If $\ell$ is much smaller than the true local sensitivity then we may truncate too much of the data. 
%
\squishend
}

\section{Experiments}\label{sec:experiments}

We evaluate the efficiency and accuracy of \TSens. Experiments are designed to answer following questions:
\squishlist
\item How tight is the local sensitivity computed by \TSens compared to other algorithms like elastic sensitivity~\cite{DBLP:journals/corr/JohnsonNS17}?
\item How does \TSens' runtime compared to that of (a) the elastic sensitivity algorithm and (b) query evaluation?
\item Does the truncation with \TSens mechanism result in more accurate differentially private query answering than prior work like \PrivSQL \cite{kotsogiannis2019privatesql}?
\squishend
\eat{
Although it is proved that \TSens achieved the tightest local sensitivity, it is unknown how much it improves the bound of local sensitivity and how much it increases the time usage in practice. What's more, since \TSens can also track tuple sensitivity as a side product, it can truncate the database to be of a fixed global sensitivity and at the meanwhile guarantee a low bias, but it is unknown to which extent it improves the utility. 
}

\revm{We use synthetic datasets from TPC-H benchmark \cite{TPCH} and real world datasets of Facebook ego-networks from SNAP~\cite{snapnets} and designed seven full conjunctive queries with different query complexities to evaluate the performance of \TSens. }These queries are also used to evaluate the performance of DP mechanism supported by \TSens. The results are compared with the sensitivity engine \Elastic from `Flex' \cite{DBLP:journals/corr/JohnsonNS17} and with the 
differentially private SQL answering engine \PrivSQL from `PrivateSQL' \cite{kotsogiannis2019privatesql}. We name our sensitivity algorithm as \TSens and its DP application as \TSensDP. 

\eat{
\TSens is implemented by 1.5k lines of python and it uses Postgres as the database engine. We extend 'Elastic Sensitivity' so it can compute the sensitivity of a cross product. As to the DP experiment, we also borrow the concept of policy and primary private relation from 'PrivateSQL' and assume the query has already been rewritten according to the policy. 
}

A summary of our key findings: 
\squishlist
	\item \TSens achieves at most 2,200,000 times smaller local sensitivity compared to \Elastic for a simple cyclic query for a database with 866,602 tuples.
	\item \TSens has \revb{on average} 80\% - 320\% overhead compared to query evaluation for different queries. It is 2 - 60 times slower than \Elastic, but returns a local sensitivity value that is 6 - 60,000 times smaller \revb{on average}.
	\item \revm{\PrivSQL has more than 99\% relative error (almost worse than just returning 0 as the answer) for four of the seven queries. \TSensDP answers 8 queries with $\leq 8\%$ relative error and the last query with $\leq 20\%$ relative error.}
\squishend

\subsection{Setup}

\begin{figure}[b]
\begin{subfigure}[b]{0.23\textwidth}
    \includegraphics[width=\linewidth]{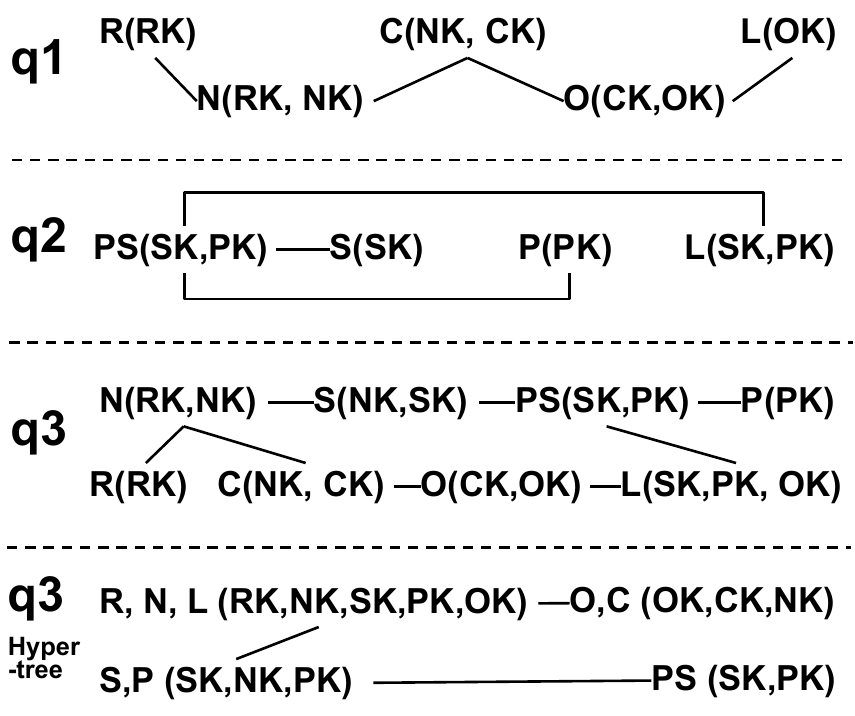}
    \caption{TPC-H queries}
    \label{fig:tpch_queries}
    \vspace{-5pt}
\end{subfigure}
\begin{subfigure}[b]{0.23\textwidth}
    \includegraphics[width=\linewidth]{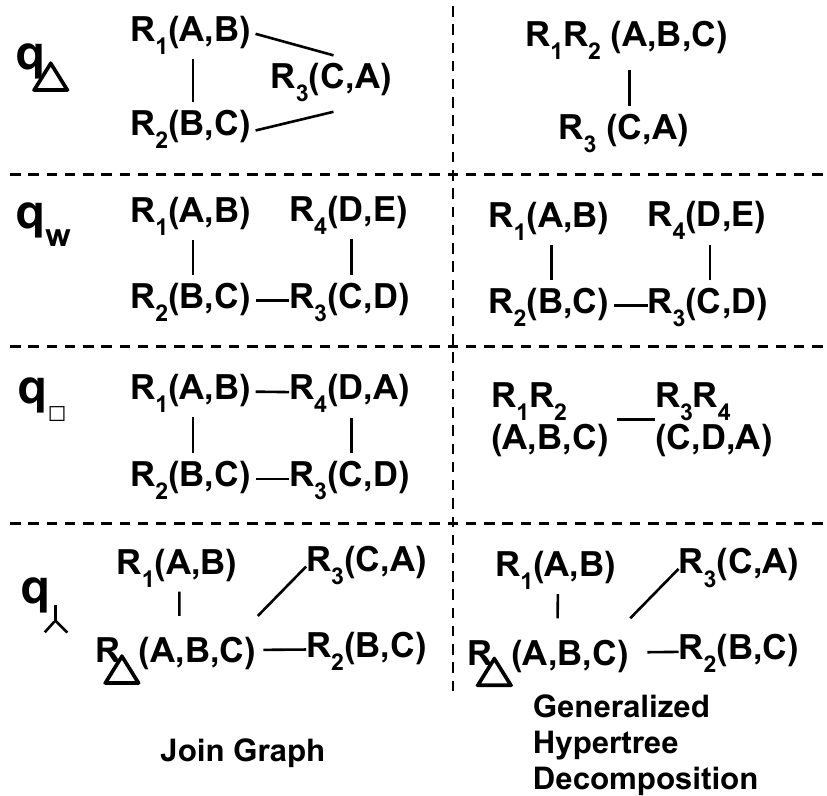}
    \caption{\revm{Facebook queries}}
    \label{fig:facebook_queries}
    \vspace{-5pt}
\end{subfigure}%
\vspace{-5pt}
\caption{\change{The join plan for each query.} 
}
\label{fig:queries}
\vspace{-10pt}
\end{figure}

\eat{
\begin{figure*}[t!]
	\centering
	\begin{subfigure}[b]{0.15\textwidth}
		\centering
		\includegraphics[width=\textwidth]{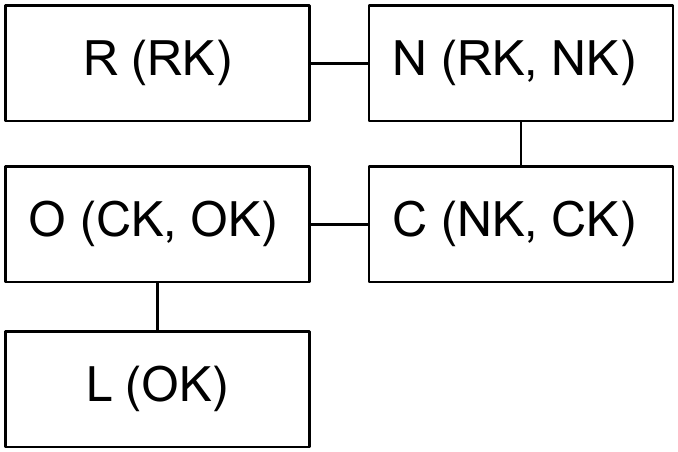}
		\caption{q1 join plan}
		\label{fig:tpch_q1_graph}
	\end{subfigure}~~~
	\begin{subfigure}[b]{0.14\textwidth}
		\centering
		\includegraphics[width=\textwidth]{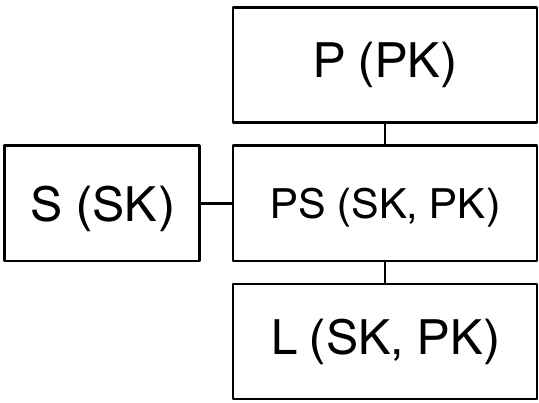}
		\caption{q2 join plan}
		\label{fig:tpch_q2_graph}
	\end{subfigure}~~~
	\begin{subfigure}[b]{0.25\textwidth}
		\centering
		\includegraphics[width=\textwidth]{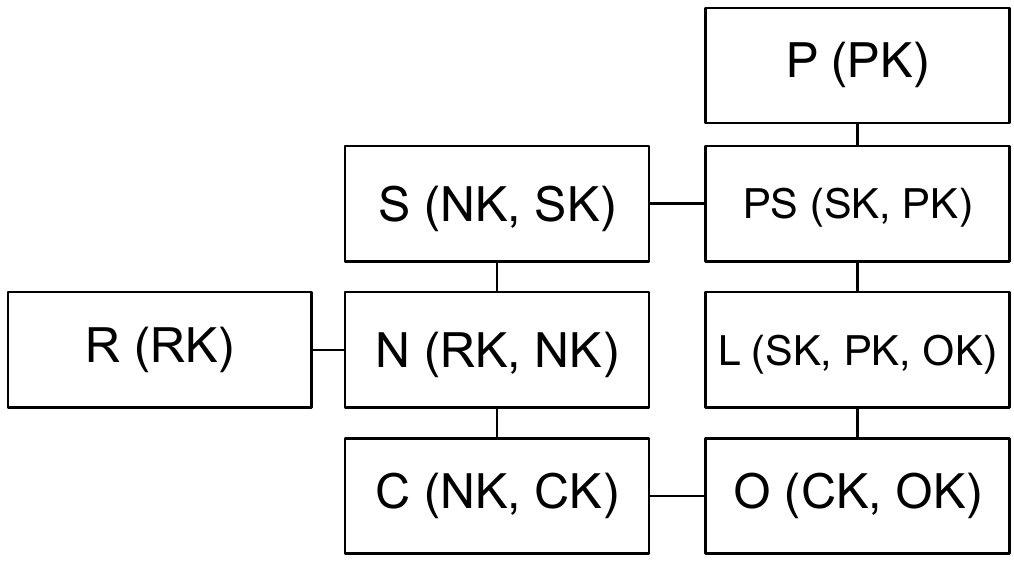}
		\caption{q3 join graph}
		\label{fig:tpch_q3_graph}
	\end{subfigure}~~~
	\begin{subfigure}[b]{0.25\textwidth}
		\centering
		\includegraphics[width=\textwidth]{{"TPSens_-_q3_-_join plan"}.pdf}
		\caption{q3 join plan}
		\label{fig:tpch_q3_tree}
	\end{subfigure}
	\caption{TPC-H query graphs and join plans 	}
	\vspace{-10pt}
\end{figure*}
}


\paratitle{Dataset}. \revm{We evaluate our algorithms on synthetic TPC-H datasets~\cite{DBGEN} and real world Facebook dataset~\cite{snapnets}.

\emph{TPC-H}. \change{We consider synthetic datasets generated from TPC-H benchmark~\cite{DBGEN} with the following schema:\\
Attributes: {\small \texttt{RegionKey}(RK), \texttt{NationKey}(NK), \texttt{CustKey}(CK), \texttt{OrderKey}(OK), \texttt{SuppKey}(SK), \texttt{PartKey}(PK)} \\
Relations:
{\small \texttt{Region}(R:RK), \texttt{Nation}(N:RK,NK), \texttt{Customer}(C:NK,CK), \texttt{Orders}(O:CK,OK), \texttt{Supplier}(S:NK,SK), \texttt{Part}(P:PK), \texttt{Partsupp}(PS:SK,PK), \texttt{Lineitem}(L:OK,SK,PK)}.
}

We evaluate the scalability of our algorithm on TPC-H datasets at different scales $\{0.0001, 0.001, 0.01, 0.1, 1, 2, 10\}$. 
At scale 1, the sizes of for these relations are 5, 25, 1e4, 1.5e5, 2e5, 8e5, 1.5e6, 6e6 respectively. 
The same schema and datasets were used to evaluate prior work on differentially private SQL query answering \cite{DBLP:journals/corr/JohnsonNS17, kotsogiannis2019privatesql}.
}

\revm{\emph{Facebook}. We use the Facebook ego-networks from SNAP (Stanford Network Analysis Project)~\cite{snapnets}. An ego-network of a user is a set of ``social cirles'' formed by this user's friends~\cite{leskovec2012learning}.
This dataset consists 10 ego-networks, 4233 circles, 4039 nodes and 88234 edges. We choose the ego-network of user 348 who has 567 circles, 225 nodes and 6384 edges, create edge tables $E_i(x,y)$ for each circle $i$ such that both users of each edge is from the circle $i$ and sort them by table size in descending order. We further create tables $R_1(x, y)$, $R_2(x, y)$, $R_3(x, y)$, $R_4(x, y)$ and insert $E_j$ into $R_i$ if the rank of $E_j$ mod 4 = i. We also create a 3-column table $R_{\triangle}(x, y, z) :- R_4(x,y), R_4(y, z), R_4(z, x)$ as a triangle table. All edges are bi-directed.} 


\begin{figure*}[t!]
	\centering
		\begin{subfigure}[b]{0.44\textwidth}
		\centering
        	\includegraphics[width=.8\textwidth, trim={0 0 10 0}, clip]{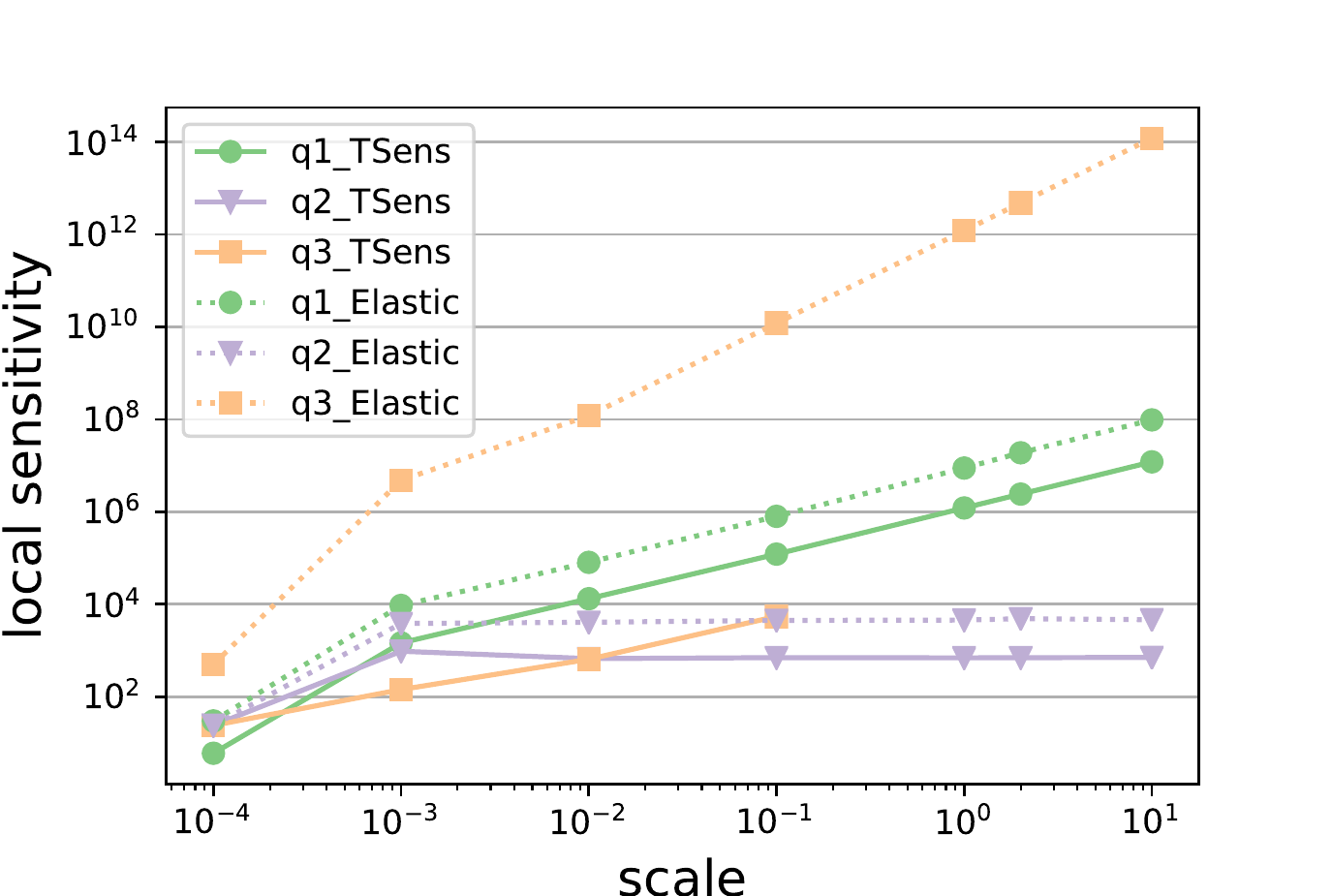}
        	\caption{Local sensitivity reported by \TSens and \Elastic for queries q1, q2 and q3 on datasets of differing scales }
        	\label{fig:sens_change_by_n}
        \end{subfigure} \ 
        \begin{subfigure}[b]{0.52\textwidth}
        	\centering
        	\includegraphics[trim={45, 120, 45, 50}, clip, width=.8\textwidth]{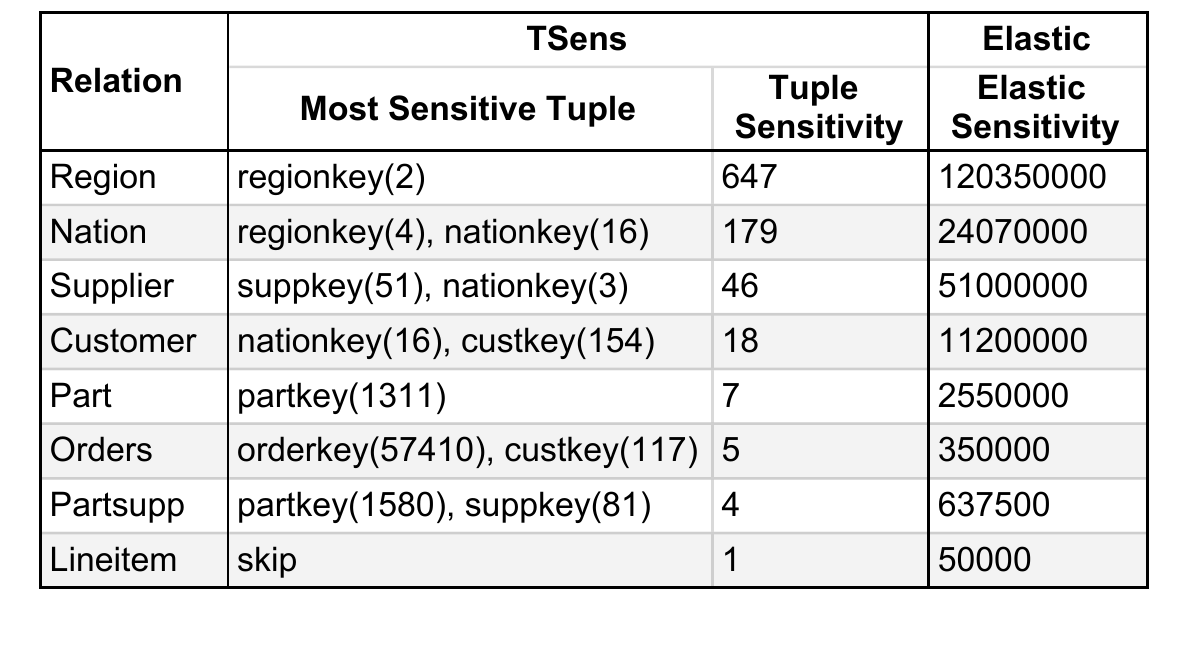}
        	\vspace{-10pt}
        	\caption{Most sensitive tuples and their tuple sensitivities for each relation of q3 when scale = 0.01. 
        	}
        	\label{fig:sens_tuple}
        \end{subfigure}
   	\vspace{-10pt}
    \caption{Local sensitivity reported by \TSens versus \Elastic for TPC-H queries} 
    \vspace{-5pt}
\end{figure*}

\begin{figure*}[t!]
	\centering
	\includegraphics[width=\textwidth, trim={10 16 10 10}, clip]{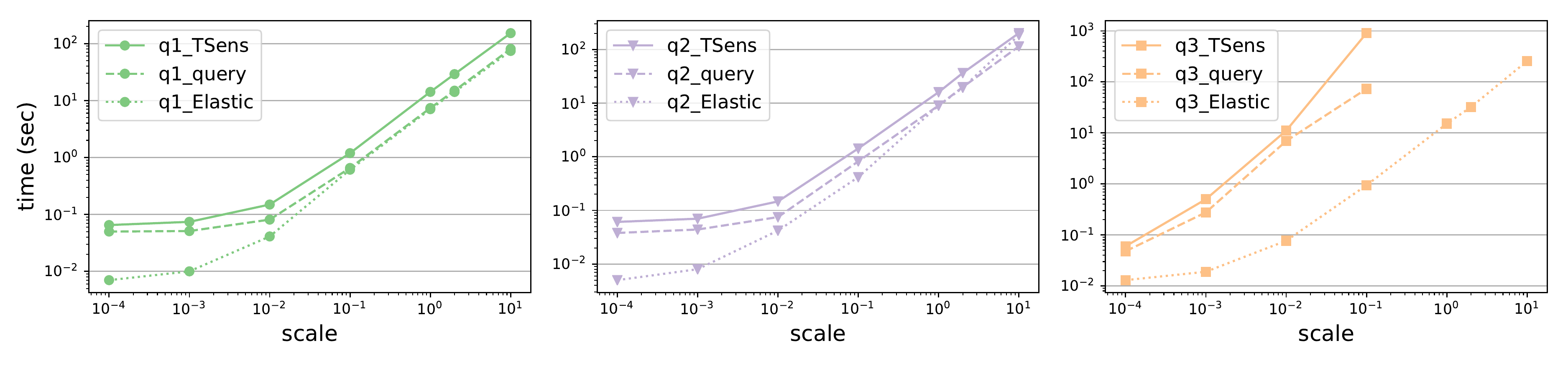}
	\vspace{-20pt}
	\caption{The trend of time usage in terms of various scales for queries q1, q2 and q3 and algorithms \TSens and \Elastic for TPC-H queries. The base line `query' shows the query evaluation time.}
	\label{fig:time_change_by_n}
	\vspace{-10pt}
\end{figure*}

\paratitle{Queries}. 
\change{We consider 3 TPC-H queries and show their query plan in Figure \ref{fig:tpch_queries} which include a path join query q1, an acyclic join query q2, and a cyclic join query q3. The third query q3 is a cyclic join query that builds a universal table with an extra constraint that the supplier and customer should be from the same nation.} \revm{We also consider 4 Facebook queries as shown in Figure \ref{fig:facebook_queries} including a triangle query  $\qtri(A,B,C)$, a path join query  $\qpath(A,B,C,D,E)$, a 4-cycle query $\qcycle(A,B,C,D)$, and a star join query $\qstar(A,B,C)$. We also show the generalized hypertree decomposition for all non-acyclic queries in the same figure. }


We use a machine with 2 processors, 512G SSD and 16G memory to run experiments. Each query is repeated 10 times.

\eat{We consider 3 queries for the TPC-H dataset as Figure \ref{fig:tpch_queries} shows:
\squishlist
\item[(q1)] A path join query  $q1(RK,NK,CK,OK)$ $:-$  $R(RK),$ $N(RK,NK),$ $C(NK,CK),$ $O(CK,OK),$ $L(OK)$ that goes through R, N, C, O and L.
\item[(q2)] An acyclic join query $q2(SK,PK) :- S(SK), P(PK), PS(SK,PK), L(SK,PK)$ whose root is PS and leaf relations are P, S and L. 

\item[(q3)] A simple cyclic join query $ q3(RK,NK,CK,OK,SK,PK)$ $:-$ $R(RK),$ $N(RK,NK),$ $C(NK,CK),$ $O(CK,OK),$ $S(NK,SK),$ $P(PK),$ $PS(SK,PK),$ $L(SK,PK)$ that combines q1 and q3 together. This query builds a universal table with an extra constraint that the supplier and customer should be from the same nation.
\squishend
}

\eat{
\begin{figure}
	\centering
	\begin{subfigure}[b]{0.2\textwidth}
		\centering
		\includegraphics[width=\textwidth]{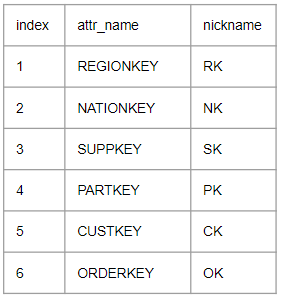}
		\caption{TPC-H Attributes}
		\label{fig:tpch_attr}
	\end{subfigure}
	\begin{subfigure}[b]{0.25\textwidth}
		\centering
		\includegraphics[width=\textwidth]{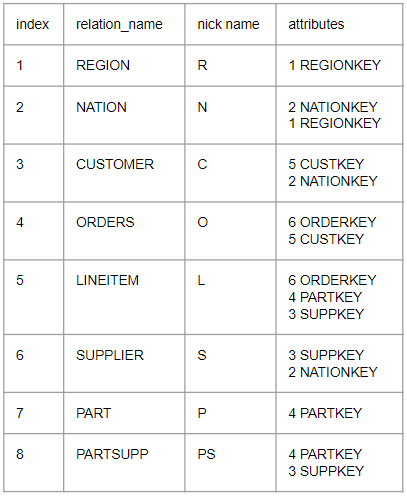}
		\caption{TPC-H Tables}
		\label{fig:tpch_table}
	\end{subfigure}
	\caption{TPC-H Schema}
\end{figure}

\begin{figure}
	\centering
	\includegraphics[width=0.5\textwidth]{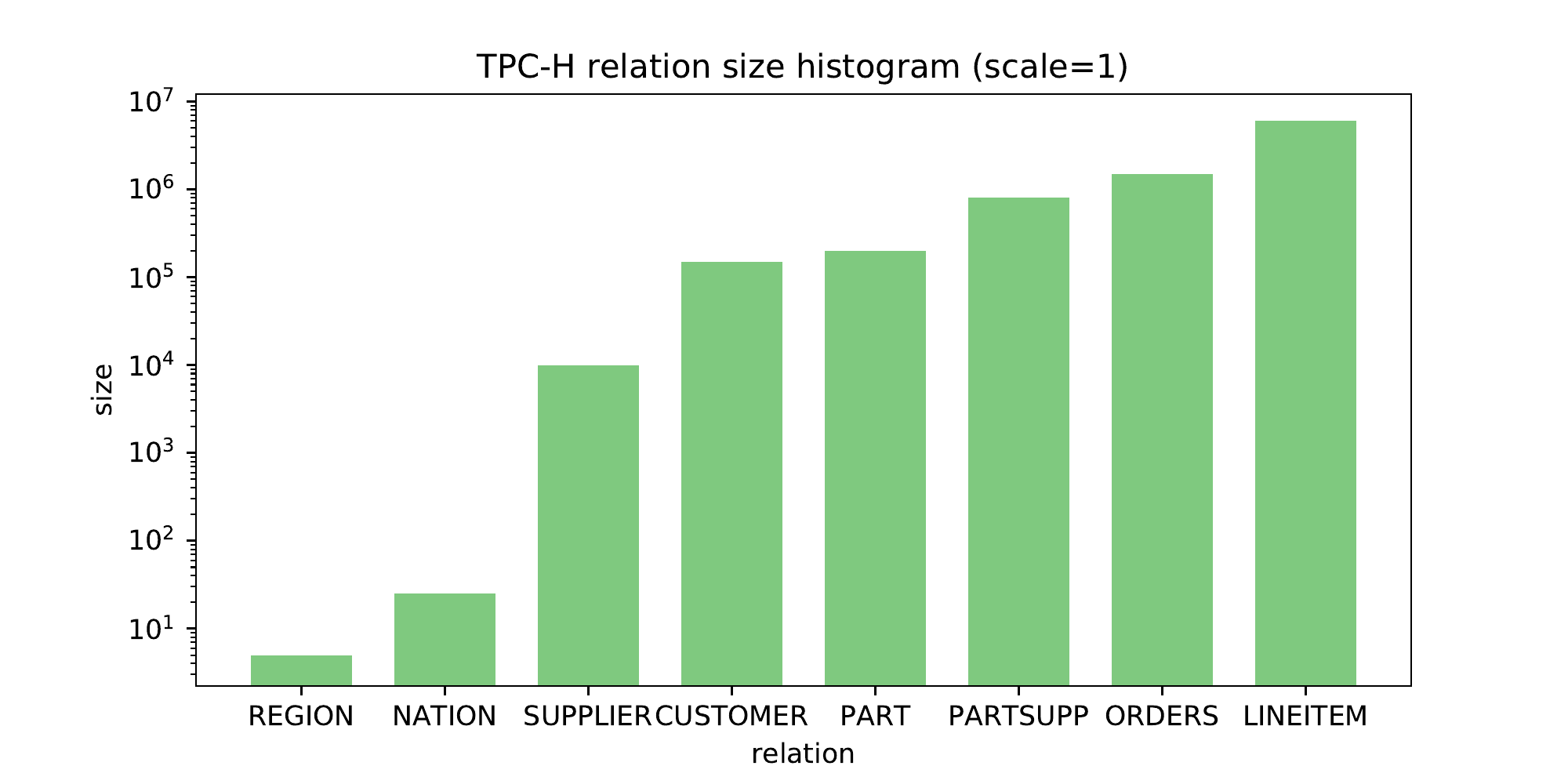}
	\caption{TPC-H Relation Size Histogram}
	\label{fig:tpch_dbsize}
\end{figure}
}

\vspace{-1em}
\subsection{Local Sensitivity}\label{sec:exp_ls}
\noindent\textbf{Baseline}.
\change{We compare the accuracy and runtime of our \TSens algorithm with prior technique \Elastic~\cite{DBLP:journals/corr/JohnsonNS17} for finding the local sensitivity of a given query. As the original \Elastic algorithm  requires the maximum frequency of the join attributes to derive the upper bound of the local sensitivity, we first let \Elastic pre-process the database to obtain the max frequency for its sensitivity analysis.} We also extend \Elastic algorithm to support cross-product by assigning the max frequency of empty attributes as the size of the table and to take the join plan as input  so that the join order in the experiment is the same. We define the join order as a post-traversal of the join plan.

\revm{We also compare the algorithm runtime to the query evaluation time. We apply Yannakakis algorithm to compute the size of query output. For queries that are not acyclic, we first compute the join for each node in the generalized hypertree, and then apply Yannakakis algorithm. The time of running \Elastic is also reported. We run each algorithm 10 times to report the average time.}

\paratitle{Result and Analysis}. Figure \ref{fig:sens_change_by_n} shows the local sensitivity trend in aspect to the scale for TPC-H. Notice that after scale 0.001, \TSens has \revb{on average} 7x smaller and 6x smaller of the local sensitivity for q1 and q2 \revb{than} \Elastic has. Moreover, for q3, \TSens achieves 2,200,000x smaller value for the local sensitivity than \Elastic does when scale equals to 0.1. We didn't run q3 for scale larger than 0.1 due to the memory limit issue. 
\revb{The multiplicity tables for this cyclic query grows nearly quadratically with the input table size. Our future work will extend our algorithm to maintain the top $k$ frequent values instead of all the frequencies which can reduce the intermediate size and further speed up runtime (Section~\ref{sec:extensions}).}

Figure \ref{fig:sens_tuple} shows the most sensitive tuple found by \TSens for each relation for q3 at scale = 0.01. \change{Unlike \TSens, \Elastic can only obtain a local sensitivity upper bound, but cannot find the most sensitive tuple. Hence, we report the most sensitive tuple for each relation while also reports its elastic sensitivity by setting this relation as the only sensitive table for \Elastic.} Each tuple sensitivity \change{found by \TSens} is below 1,000 while the least elastic sensitivity \change{reported by \Elastic} is beyond 10,000. We skip finding the most sensitive tuple in LINEITEM since it has the superkey in the query head and thus the tuple sensitivity is at most 1.



Figure \ref{fig:time_change_by_n} shows the time cost for both \TSens and \Elastic for different queries and scales for TPC-H. The second line $q_i\_query$ is the query evaluation time. \revm{Notice that we skip computing the multiplicity table of Lineitem in q3 since the tuple sensitivity is at most 1 due to FK-PK joins.} For q1 and q2, both \TSens and \Elastic shows a tight relation to the time of query evaluation, which takes \revb{on average} 1.8x and 0.9x of the query evaluation time after scale 0.001. For q3, although \Elastic is much faster, \TSens only takes \revb{on average} 4.2x of the query evaluation time to find \revb{on average} 60,000x smaller local sensitivity before scale 1.

\revm{We also report the accuracy and the runtime of \TSens and \Elastic for Facebook queries in Table~\ref{table:lsfacebook}. The sensitivity bound improvement ranges from $\times 3$ to $\times 80k$. Although \TSens spends $\times 25$ to $\times 60$ more time than \Elastic, its runtime is comparable to query evaluation time. The local sensitivity can also be computed by repeating query evaluation over databases  which are formed by removing a tuple from active domain or inserting a tuple from representative domain one at a time (using  a variation of \cite{Yannakakis:1981} as discussed in Sections~\ref{sec:pj_simple_ineffi}  and \ref{sec:general-algo}). However, the size of the active domain and representative domain is above 10k. This approach will take $\times 10k+$ time than \TSens.} 


\eat{
\begin{figure}[b]
\centering
	\includegraphics[width=.4\textwidth, trim={0 0 10 0}, clip]{sens_change_by_n.pdf}
	\caption{Local sensitivity reported by \TSens and \Elastic for queries q1, q2 and q3 on datasets of differing scales }
	\label{fig:sens_change_by_n}
\end{figure} 
\begin{figure}[b]
	\centering
	\includegraphics[trim={45, 120, 45, 50}, clip, width=.4\textwidth]{{"figures/TPSens - Experiment - MostSensitiveTuple2"}.pdf}
	\caption{Most sensitive tuples and their tuple sensitivities for each relation of q3 when scale = 0.01. 
	}
	\label{fig:sens_tuple}
\end{figure}
}

\begin{table}[t]
\def\tsens{{\bfseries  \TSens}}
\def\elastic{{\bfseries  \Elastic}}
\hspace{-5pt}
\scalebox{0.9}{
\setlength\tabcolsep{5pt}
\begin{tabular}{|l|l|l|l|l|l|l|l|l|}
\hline
& \multicolumn{2}{c|}{\bfseries Local Sensitivity} & \multicolumn{3}{c|}{\bfseries Time (seconds)} \\ \cline{2-6}
& \tsens & \elastic & \tsens & \elastic & \Evaluation \\ \hline
$\qtri$ & \cellcolor{gray!20}87& 7524& 0.405& 0.007& 0.431\\ \hline
$\qpath$ & \cellcolor{gray!20}178923& 511632& 0.237& 0.010& 0.182\\ \hline
$\qcycle$ & \cellcolor{gray!20}2014& 511632& 0.618& 0.009& 0.465\\ \hline
$\qstar$ & \cellcolor{gray!20}34& 2723688& 0.604& 0.012& 0.175\\ \hline
\end{tabular}
}
\caption{\revm{Local sensitivity and runtime of 4 query types for \TSens and \Elastic for Facebook queries. It also reports the query evaluation time for counting the output size. Gray cells have tighter local sensitivities. }} \label{table:lsfacebook}
    \vspace{-20pt}
\end{table}

\eat{

\begin{table}[h]
\small
\def\tsensdp{TSensDP}
\def\privatesql{PrivateSQL}
\scalebox{0.9}{
\begin{tabular}{|l|r|l|r|r|r|r|r|l|}
\hline
 {\bfseries  \makecell{\scriptsize Qu- \\\scriptsize ery}} & {${|Q(D)|}$} & {\bfseries Algorithm} & \multicolumn{1}{c|}{\bfseries Error} & \multicolumn{1}{c|}{\bfseries Bias} & \multicolumn{1}{c|}{\bfseries \makecell{\scriptsize Global \\\scriptsize Sens.}} & \multicolumn{1}{c|}{\bfseries  Time}\\ \hline
$\qtri$ &30699 & \cellcolor{gray!20}TSensDP & \cellcolor{gray!20}1.50\%  & 1.47\% & 49 & 0.562\\ \cline{3-7}&& PrivateSQL & 19.12\%  & 0.00\% & 6732 & 0.230\\ \cline{3-7}\hline
$\qpath$ &17555419 & TSensDP & 5.59\%  & 5.69\% & 17440 & 0.843\\ \cline{3-7}&& \cellcolor{gray!20}PrivateSQL & \cellcolor{gray!20}2.25\%  & 0.00\% & 289476 & 10.340\\ \cline{3-7}\hline
$\qcycle$ &142903 & \cellcolor{gray!20}TSensDP & \cellcolor{gray!20}2.00\%  & 1.77\% & 167 & 0.792\\ \cline{3-7}&& PrivateSQL & 127.16\%  & 0.00\% & 289476 & 2.232\\ \cline{3-7}\hline
$\qstar$ &786 & \cellcolor{gray!20}TSensDP & \cellcolor{gray!20}19.02\%  & 16.16\% & 13 & 0.670\\ 
\cline{3-7}&& PrivateSQL & 498K\%  & 0.00\% & 2436K & 0.290\\ 
\cline{3-7}\hline\end{tabular}
}
\caption{\revm{Application to DP: Comparison between \TSensDP and \PrivSQL for Facebook queries. Time is in seconds. Gray cells achieve lower errors.}}
    \vspace{-20pt}
\end{table}
}

\begin{table}[h]
\small
\def\tsensdp{TSensDP}
\def\privatesql{PrivateSQL}
\hspace{-5pt}
\scalebox{0.85}{
\setlength\tabcolsep{5pt}
\begin{tabular}{|l|r|l|r|r|r|r|r|l|}
\hline
 {\bfseries  \makecell{\scriptsize Qu- \\\scriptsize ery}} & {${|Q(D)|}$} & {\bfseries Algorithm} & \multicolumn{1}{c|}{\bfseries Error} & \multicolumn{1}{c|}{\bfseries Bias} & \multicolumn{1}{c|}{\bfseries \makecell{\scriptsize Global \\\scriptsize Sens.}} & \multicolumn{1}{c|}{\bfseries  Time}\\ \hline
q1 &60175 & TSensDP & 3.56\%  & 3.44\% & 119 & 0.693 \\ \cline{3-7}&& 
\cellcolor{gray!20} PrivateSQL & \cellcolor{gray!20} 1.34\%  & 1.02\% & 220 & 0.292\\ \cline{3-7}\hline 
q2 &60175 & \cellcolor{gray!20}TSensDP & \cellcolor{gray!20}7.71\%  & 7.62\% & 640 & 0.554\\ \cline{3-7}&& PrivateSQL & 99.03\%  & 100.00\% & 774 & 0.231\\ \cline{3-7}\hline
q3 &2333 & \cellcolor{gray!20}TSensDP & \cellcolor{gray!20}2.84\%  & 0.00\% & 14 & 23.063\\ \cline{3-7}&& PrivateSQL & 1293\%  & 2.14\% & 12375k & 0.546\\ \cline{3-7}\hline
$\qtri$ &30699 & \cellcolor{gray!20}TSensDP & \cellcolor{gray!20}1.50\%  & 1.47\% & 49 & 0.562\\ \cline{3-7}&& PrivateSQL & 19.12\%  & 0.00\% & 6732 & 0.230\\ \cline{3-7}\hline
$\qpath$ &17555419 & TSensDP & 5.59\%  & 5.69\% & 17440 & 0.843\\ \cline{3-7}&& \cellcolor{gray!20}PrivateSQL & \cellcolor{gray!20}2.25\%  & 0.00\% & 289476 & 10.340\\ \cline{3-7}\hline
$\qcycle$ &142903 & \cellcolor{gray!20}TSensDP & \cellcolor{gray!20}2.00\%  & 1.77\% & 167 & 0.792\\ \cline{3-7}&& PrivateSQL & 100\%  & 0.00\% & 289476 & 2.232\\ \cline{3-7}\hline
$\qstar$ &786 & \cellcolor{gray!20}TSensDP & \cellcolor{gray!20}19.02\%  & 16.16\% & 13 & 0.670\\ 
\cline{3-7}&& PrivateSQL & 30K\%  & 0.00\% & 2437K & 0.290\\ 
\cline{3-7}\hline\end{tabular}
}
\caption{\revm{Application to DP: Comparison between \TSensDP and \PrivSQL for TPC-H and Facebook queries. Time is in seconds. Gray cells achieve lower errors.}} \label{table:dpapp}
    \vspace{-5mm}
\end{table}

\subsection{Differential Privacy}\label{sec:exp_dp}

\paratitle{Baseline}. 
\change{\PrivSQL is a differentially private SQL answering engine and it introduces the concept of policy that given primary private relations, the sensitivity of other related relations should be updated to be non-zero according to the database key constraints. For TPC-H datasets,} we consider \texttt{CUSTOMER} is the primary private relation for q1 and q3, and \texttt{SUPPLIER} is the primary private relation for q2, so the sensitivity of \texttt{ORDERS} is affected by \texttt{CUSTOMER} and the sensitivity of \texttt{PARTSUPP} is affected by \texttt{SUPPLIER}. The sensitivity of \texttt{LINEITEM} is affected by either of them. \revm{For Facebook dataset, we consider $R_2$ is the primary private relation. 
}


\PrivSQL uses the maximum frequency as the truncation threshold, which is different from using tuple sensitivity to truncate the database. Any tuple whose frequency is beyond the max frequency will be dropped from the database. \PrivSQL runs SVT to learn the truncation threshold for each relation; however, the noise scale of SVT depends on the sensitivity of the relation while it is constantly 1 in \TSensDP. 

Although the privacy budget allocation strategy affects the performance of DP algorithm, we skip exploring this effect and assume \PrivSQL and \TSensDP divide the privacy budget into two halves, one for the threshold learning and the other for reporting the query result after truncation. We disable the synopsis generation phase of \PrivSQL and just use Laplace mechanism to answer the SQL query directly.

\eat{
\begin{figure*}[t]
	\centering
	\includegraphics[trim={25, 80, 25, 50}, clip, width=.9\textwidth]{{"figures/TPSens - DP - Sheet1"}.pdf}
	\vspace{-10pt}
	\caption{Application to Differential Privacy: Comparing \TSensDP error to \PrivSQL error for TPC-H queries}
	\label{fig:dp}
	\vspace{-10pt}
\end{figure*}
}

\paratitle{Result and Analysis}. 
\revm{Table~\ref{table:dpapp} shows the statistics of releasing differential private query results by \TSensDP or \PrivSQL for TPC-H and Facebook datasets. Output below 0 is truncated to 0. We report the median of global sensitivity, the median of relative absolute bias, the median of relative absolute error and the average time for each query over 20 runs.}
We assume the table size is given.
For TPC-H, we assume the maximum tuple sensitivity of q1 is 100, of q2 is 500 and of q3 is 10. 
\revm{\TSensDP has $\leq4$\% error for $q_1$ and $q_3$, and $\leq 8$\% error for $q_2$. In contrast, \PrivSQL has more than $99$\% error on $q_2$ and $q_3$.} This means that the error in \PrivSQL answers for these queries is worse than returning 0 as the answer without looking at the data. The reasons for the poor error are different. In $q_2$ \PrivSQL truncates too much of the data, while in $q_3$ it estimate a very loose bound on sensitivity.  

\revm{
For Facebook dataset, we assume the maximum tuple sensitivity of $\qtri$ is 70, of $\qpath$ is 25k, of $\qcycle$ is 200 and of $\qstar$ is 15. 
\TSensDP achieves $< 6\%$ error for $\qtri$, $\qpath$, $\qcycle$, while \PrivSQL get $> 100\%$ error for $\qcycle$ and $\qstar$. Since there is no FK-PK join for Facebook queries, we have only one primary private table, which means no table truncation and thus has 0 bias in \PrivSQL. However, \PrivSQL has $\times 10$ to $\times 180k$ larger global sensitivity than \TSensDP, which dominates the error. }

\revb{\paratitle{Parameter Analysis}. To find how 
the upper bound parameter for tuple sensitivity 
$\ell$ affects the performance, we vary $\ell$ through 1, 10, 30, 50, 100, 1000 and repeat \TSensDP 20 times for the star query $\qstar(A,B,C)$ whose true local sensitivity is 13 when $R_2$ is the primary private relation for DP.
For each bound, the median global sensitivity, which is also the tuple sensitivity threshold learned from the SVT routine, is [11, 13, 9, 4, 48, 160], the median bias error is [3\%, 1\%, 13\%, 55\%, 0\%, 0\%], and the median relative error is [5\%, 4\%, 17\%, 56\%, 32\%, 98\%]. The optimal $\ell$ in this case is 10, with the corresponding error as 4\%, while the worst error is 98\% when $\ell = 1000$. 

Notice that as $\ell$ increases, the noise added to $\hat{\CQ}$ in the SVT routine gets larger. This causes the learned tuple sensitivity threshold to deviate more from the local sensitivity, which is considered the optimal threshold by the rule of thumb. When $\ell$ is too small, the learned tuple sensitivity threshold could also be small, which increases the bias.}




\balance
\vspace{-0.2cm}
\section{Related Work}\label{sec:related}

Sensitivity analysis for SQL queries is important to the design of differentially private algorithms. The focus of existing work~\cite{DBLP:journals/corr/abs-1207-0872,Laud2018SensitivityAO,arapinis16,DBLP:journals/jpc/EbadiS16, kotsogiannis2019privatesql,mcsherry2009privacy} is to compute the global sensitivities of SQL queries or their upper bounds. The earliest work by McSherry along this line~\cite{mcsherry2009privacy} applies static analysis on a given relational algebra and then combines the sensitivities of the operators in the relational algebra to obtain the maximum possible change to the query output for all possible database instances. 
This analysis is independent of the database instance, so the result can be much larger than the local sensitivity. In particular, for join operator, the global sensitivity can be unbounded. The analysis either considers a restricted form of join~\cite{mcsherry2009privacy, DBLP:journals/jpc/EbadiS16} or  constrained database instances~\cite{ DBLP:journals/corr/abs-1207-0872,Laud2018SensitivityAO,arapinis16}. For general join queries on unconstrained databases, Lipschitz extension~\cite{NodeDP:TCC2013, blocki2013differentially, DBLP:conf/sigmod/ChenZ13, kotsogiannis2019privatesql} is usually applied to transform the original query $Q$ that has an unbounded sensitivity into a different query $Q'$ that (a) has abounded global sensitivity and (b) has a similar answer as $Q$. In particular, the transformed query in PrivateSQL~\cite{kotsogiannis2019privatesql} require to truncates the sensitive tuples. Hence, our work offers efficient ways to identify the most sensitive tuples to complement PrivateSQL to achieve differential privacy. 

Smooth sensitivity~\cite{DBLP:conf/stoc/NissimRS07,johnson2018towards} is another important sensitivity notion for achieving differential privacy. This sensitivity is a smooth upper bound of the local sensitivity of databases at a distance from the given database instance. This requires the computation of local sensitivity of exponentially number of database instances. For SQL queries, elastic sensitivity~\cite{johnson2018towards} provides efficient static analysis rule to estimate the upper bound of local sensitivity, but this bound can be still very loose. For example, even if the local sensitivity for a query with selection operator is small, the elastic sensitivity algorithm will output the same value as for a query without the selection operators. In addition, the computation of elastic sensitivity requires additional constraint cardinality information of the given database instance. 

Smooth sensitivity~\cite{DBLP:conf/stoc/NissimRS07,Zhang:2015:PRG:2723372.2737785} or Lipschitz extension~\cite{DBLP:conf/sigmod/ChenZ13,NodeDP:TCC2013,blocki2013differentially} have been mainly applied to release graph statistics. However, these algorithms either require customized analysis for each new query~\cite{Zhang:2015:PRG:2723372.2737785} or suffer from high performance cost~\cite{DBLP:conf/stoc/NissimRS07,DBLP:conf/sigmod/ChenZ13,NodeDP:TCC2013,blocki2013differentially}. We will extend our study to graph queries (involving self-joins) in the future.

Sensitivity analysis has also been studied for non-SQL functions~\cite{
Gaboardi:2013:LDT:2480359.2429113, 
Reed:2010:DMT:1863543.1863568, 
Chaudhuri:2011:PPR:2025113.2025131
}, with a focus on global sensitivity. 
Related topics also include sensitivity analysis for probabilistic queries \cite{Kanagal:2011:SAE:1989323.1989411} and finding responsibility of tuples \cite{MeliouGMS11}, where the goal is different from ours.
Prior work on provenance for queries and deletion  propagation (e.g., \cite{AmsterdamerDT11, Buneman:2002:PDA:543613.543633}) provide analysis for a rich set of queries and explanations for query results, but the analysis is mainly for removing a tuple in the database (downward tuple sensitivity). Our work also considers upward tuple sensitivity which involves adding new tuples from the domain. Our future study will consider general aggregates and functions. 


\eat{
\subsection{Limitations in Prior Work}
We present several examples where prior work \cite{Ebadi:2015:DPG:2676726.2677005,DBLP:journals/corr/JohnsonNS17} fail.

\subsubsection{Limitations in ProPer}
The ProPer system ~\cite{Ebadi:2015:DPG:2676726.2677005} aims to provide provenance for the personalized privacy of individuals with records in the system. The system maintains a privacy budget for each individuals in the system. As queries are performed over time, the budget for each individual may decrease. The system consider two types of queries: transformation query and primitive differentially private query. A transformation query includes set operations, record selection with where and projection. For each transformation query, the system stores the result of this query as a \emph{table variable}.
A primitive differentially private query simply corresponds to answering the size of the variable table (the number of records) with $\epsilon$ privacy budget. For such a query, the system first drops the records in the variable table with insufficient privacy budget $<\epsilon$, and then outputs the sum of the number of remaining records in the variable table and the noise drawn from Laplace distribution with budget $\epsilon$. Only the privacy budgets of the individuals associated these remaining records get reduced by $\epsilon$.

{\bf Unsupported queries.}
ProPer only supports selection with where, projection, and one level of aggregation, not join or aggregations over aggregations.

{\bf No utility guarantee.}
The paper for ProPer also admits that this approach does not have utility guarantee. Suppose most of records in the true query result have insufficient privacy budget, then the outputted count will be significantly different from the true result. The paper proposed a case when the rate at which new data enters the database is sufficiently high relative to the rate at which queries consume the budget. This case may be also similar to the setting where individuals can refresh their privacy budgets over a period of time. As long as the privacy budget of most individuals is greater than the privacy budget of the queries, the bias can be small. However, as the bias introduced by dropping records is still unknown, it is hard to derive utility bounds.

\eat{
\subsubsection{Privacy leakage}
The paper assumes that the privacy budget of individuals is a private information, unknown to data analysts. Privacy attacks can be constructed to learn the  privacy budget set by individuals.
For example, suppose the adversary knows that individuals with a certain disease  are highly cautious about their privacy, and hence they set their privacy budget very small (a type of correlation between attributes of individuals and the privacy budget).
The adversary can start a counting query about this group of individuals with a small privacy budget. If the output answer is close to 0, then the adversary can gradually increase the privacy budget of the same query till a noisy answer that is much larger than 0 is reported.
Then the adversary can infer the privacy budget set by this special group of individuals.

If the privacy budget of individuals is public (assuming there is no correlation between user's privacy preference and their data values), it is not hard to construct privacy attacks (exclusion attack).
For example, suppose all individuals start with the same budget 2.0. If the first query requests the count of individuals with cancer to be answered with budget of 1.9, then the adversary knows that only the privacy budgets of individuals with cancer are reduced to 0.1.  If the analyst sends a second query which simply requests the number of individuals in the database to be answered with privacy budget 2.0. Suppose the noisy output is close 0, then the adversary can infer that almost all individuals in the database have cancer.
}

\subsubsection{Limitations in Elastic Sensitivity}
Elastic sensitivity \cite{johnson2018towards} is a method for approximating the local sensitivity of queries with general equijoins.

{\bf Unsupported queries.}
As discussed in Section 3.7 of the work\cite{johnson2018towards}, unsupported queries include (i) non-equijoins; (ii) queries do not have max-frequency metrics due to query structure, for instance, queries use counts computed in subqueries as join keys (joins over aggregates), or aggregates over aggregates. For instance, consider the query:\todo{[YC]I think this example is fine, since count(*) has the max-frequency as 1, although elastic sensitivity may not include this. I think if we add group by to the example it would be better. For example:

WITH A AS (SELECT count(*) FROM T1 GROUP BY S),
B AS (SELECT count(*) FROM T2 GROUP BY S)
SELECT count(*) FROM A JOIN B ON A.count = B.count

}
\begin{verbatim}
WITH A AS (SELECT count(*) FROM T1),
     B AS (SELECT count(*) FROM T2)
SELECT count(*) FROM A JOIN B ON A.count = B.count
\end{verbatim}

(iii) The rules in Figure 1 \todo{[YC]the Figure 1 in [3], not in this paper} does not include $mf_k(a',r_1\bowtie_{a=b} r_2,x)$, where $a'$ is different from the join key, and it does not provide evaluation of elastic sensitivity for queries like
\begin{verbatim}
SELECT T1.A FROM T1 JOIN T2 ON T1.B = T2.B
\end{verbatim}
This query has a join on column $B$, and then has a projection on column $A$. It is possible that $mf(A, T1\bowtie_{B} T2, x)=n$, but the $A$ has a multiplicity of 1. Similarly, join queries of multiple tables on different join key also have the similar issue.

(iv) The rules do not support `DISTINCT', `LIMIT', `HAVING'.

(iv) The rules do not consider constraints within or between tables in the same database.

\eat{For instance, consider $R_1(U,V )$, $R_2(U,V)$ are tables of edges representing 2 graphs. Let
$S_1(U, D)$ be the result of
`SELECT $R_1.U$, COUNT(*) FROM $R_1$ Group By $U$'
This  gives nodes and their corresponding degrees. $S_2(U,D)$ is similarly defined over $R_2$. Let the query be a join between $S_1$ and $S_2$ on the degree column $D$. The maximum frequency over the join key `degree column' is not available for computing elastic sensitivity.}

{\bf Loose upper bound.}
Elastic sensitivity does not match the best upper-bound on local sensitivity at distance $k$. For instance consider query:
\begin{verbatim}
SELECT count(*) FROM T1
JOIN T2 ON T1.B = T2.B
JOIN T3 ON T2.C = T3.C
\end{verbatim}
Suppose
\todo{[YC] I'm not familiar with the operator precedence. Do we need to use () to wrap $\{...\} \times \{...\}$, so that it becomes $(\{...\} \times \{...\})$ ?}
$$x.T_1 = \{a_1,a_2,\ldots,a_n \}\times\{b_0\} \cup \{a_0,b_1\},$$
$$x.T_2=\{(b_0,c_1)\}\cup \{b_1,\ldots,b_n \}\times c_0,$$
$$x.T_3 = \{(c_1,d_1),(c_2,d_2),\ldots, (c_n,d_n)\}.$$
This gives $$mf_1(B,T_1,x)=n,$$
$$mf_1(B,T_2,x)=1,$$
$$mf_1(C,T_2,x)=n,$$
$$mf_1(C,T_3,x)=1.$$
The elastic stability of the join between $T_2$ and $T_3$ on $C$ is at least $n$. As $mf_1(B,T_1,x)=n$,
the elastic sensitivity of the query is at least $n^2$.
However, the join output size is $n$. \todo{The join output size is n cannot indicate the tightest local sensitivity is below $n^2$}

{\bf Dependency on query plan.} Elastic sensitivity depends on query plan. We can still consider the join query above. If we  join $T_2$ and $T_3$ on $C$ first and then join the result with $T_1$, the elastic sensitivity of the query is at least $n^2$. However, if we join $T_1$ and $T_2$ on $B$ first, then the elastic stability of this join is at least $n$ and the multiplicity of $C$ in this join is 1 only. As $mf_1(C,T_3,x)=1$ and the elastic stability of $T_3$ is 1, the elastic sensitivity of the query is just $n$. Hence, these two query plans resulted in different elastic sensitivities.
}

\section{Conclusions}\label{sec:conclusion}
We studied the local sensitivity problem for counting queries with joins -- an important task for many applications like differentially private query answering and query explanations. We showed that the problem is NP-hard in combined query and data complexity even for full conjunctive queries that have an acyclic structure -- queries for which the combined complexity of query answering is PTIME. We develop algorithms for full acyclic join queries using join trees, that run in linear in the number of relations and near linear in the number of tuples for interesting sub-classes of acyclic queries including path queries and ``doubly acyclic queries", and in PTIME in combined complexity when the maximum degree in the join tree  is bounded.  Our algorithms can be extended to handle related queries that include selection predicates as well as non-acyclic queries with a certain  property  on generalized hypertree decompositions. The local sensitivity output by our algorithms is shown to be orders of magnitude tighter than prior work. Our algorithm can also be used to construct differentially private query answering methods that are more accurate than the state of the art.  Extending the framework to handle general non-acyclic queries involving self-joins, projections, negations, and other aggregate functions would be an interesting direction of future work. 

\section*{Acknowledgments}

    This work was supported by the NSF under grants 1408982, IIS-1552538, IIS-1703431; and by DARPA and SPAWAR under contract N66001-15-C-4067; and by NIH award 1R01EB025021-01; and by NSERC through a Discovery Grant.

\end{sloppypar}

\clearpage

\bibliographystyle{abbrv}
\bibliography{refs}

\clearpage


\appendix

\section{Theorem and Proofs}

\subsection{Proof for Theorem~\ref{thm:poly-data-complexity}}
\begin{proof}
\textbf{(Algorithm)} The algorithm works as follows. First, compute the maximum downward tuple sensitivity $\delta^{-*} = \max_{t \in Q(\database)}\dtsens(t, \CQ, \database)$  (see Definition~\ref{def:tuple_sens}), and note the tuple giving the max value. Next, compute the maximum upward 
tuple sensitivity as $\delta^{+*}$  = $\max_{i \in \{1, \cdots, m\}} \max_{t \in \reprdomain^{\attributes_i}}\delta^{+*}(t, \CQ, \database)$, again noting the tuple giving the maximum value. Return $\delta^{*} = \max(\delta^{+*} , \delta^{-*})$ along with the tuple that led to this highest value.  
\par
\textbf{(Correctness)} We omit the proof that this algorithm correctly computes the local sensitivity due to space constraints.
\par
\textbf{(Polynomial data complexity)} Finally we argue that the algorithm runs in time polynomial in $n = |D|$. Note that the active domain of any single attribute $A \in \attributes_D$ in any relation $R_i$ can be computed in time polynomial in $n$ (in $O(n \log n)$ time if we use sorting to remove duplicates), and $|\actdomain^{\attr, i}| \leq n$. Since each relation $R_i$ has at most $k$ attributes,  $|\reprdomain^{\attributes_i}| \leq n^k$.  Hence the above algorithm iterates over polynomial number of choices for $t$, for each $t$ it evaluates the query $Q(D \cup \{t\})$ or $Q(D - \{t\})$, which can be done in polynomial time in $n$. Hence the total time of the above algorithm is also polynomial in $n$.
\end{proof}

\subsection{Proof for Theorem~\ref{thm:combined-nphard}}
\begin{proof}
\cut{
The local sensitivity of a CQ query is defined as the maximum change to the query output size if adding to one of the tables or removing a tuple from one of the tables.
The changes to the query output size when removing a tuple is upper bounded by the changes to the query output when adding a tuple. Hence, we will just show that checking if the changes to the query output size is greater than 0 if adding a tuple to one of table is NP-hard.
}

We give a reduction from the 3SAT problem. Consider any instance of 3SAT $\phi$ with $s$ clauses
($C_1,\ldots,C_s$) and $\ell$ variables ($v_1,\ldots,v_\ell$), where each clause is disjunction three literals (a variable or its negation), and the goal is to check if the formula $\phi = C_1 \wedge \cdots \wedge C_s$ is satisfiable. 
We create an instance of the sensitivity problem $LS(Q, \database)$ with $s+1$ relations and $\ell$ attributes in total.
For each clause $C_i$ that involves variables $v_{i_1},v_{i_2},v_{i_3}$, we add a table $R_i$ with three Boolean attributes $A_{i_1},A_{i_2},A_{i_3}$, and insert all possible triples of Boolean values that satisfy the clause $C_i$ into $R_i$ in $\database$. For example, if $C_i$ is $v_{2}\lor \bar{v}_{5} \lor \bar{v}_{7}$, then $R_i(A_2, A_5, A_7)$ contains seven Boolean triples
$(0,0,0), (0,0,1),\ldots,(1,1,1)$ {\em except} $(0,1,1)$.
In addition, we create an empty relation $R_0(A_1,\ldots,A_\ell)$, which does not contain any tuple in $\database$. 
The query is:
\begin{small}
\begin{align*}
    \CQ(A_1 ,\ldots, A_\ell) = R_0(A_1, \cdots, A_\ell) \wedge \underset{i = 1, \ldots, m}{\bigwedge} \R{i}(A_{i_1},A_{i_2},A_{i_3})
\end{align*}
\end{small}
Note that $Q$ is acyclic, as all of $R_1, \cdots, R_s$ correspond to ears (see Section~\ref{sec:acyclic}). Further, the reduction is in polynomial time in the number of variables and clauses in $\phi$. Next we argue that $\phi$ is satisfiable if and only if $LS(Q, \database) > 0$.
\par
\textbf{(only if)} Suppose $\phi$ is satisfiable, and $\mathbf{v} = (v_1 = b_1, \cdots, v_\ell = b_\ell)$ is a satisfying assignment. Then the join of $R_1 \Join \cdots \Join R_s$ is not empty and $\mathbf{v}$ belongs to their join result. However, $Q(\database) = \emptyset$ as $R_0 = \emptyset$ in $\database$. Now, if we add a tuple corresponding to $\mathbf{v}$ to $R_0$, then $Q(\database \cup \{\mathbf{v}\})$ is no longer empty (at least contains $\mathbf{v}$), and therefore $LS(Q, \database) > 0$. 
\par
\textbf{(if)} 
Suppose $LS(Q, \database) > 0$. Hence there exists at least one tuple $t$ such that if it  is added to one of the relations $R_0, R_1, \cdots, R_s$, then $|Q(\database \cup \{t\})| > |Q(\database)|$. Since  $Q(\database) = \emptyset$ as $R_0$ is empty, this tuple must be inserted to $R_0$ to have  a non-empty output. Further, the projection of this tuple to $A_{i_1}, A_{i_2}, A_{i_3}$ for relation $R_i$ must match one of the existing seven tuples of $R_i$ in $\database$ to have a non-empty join result. Therefore, this tuple (Boolean values for $v_1, \cdots, v_\ell$) gives a satisfying solution for $\phi$ by satisfying all the clauses, and makes $\phi$ satisfiable.
\end{proof}

\end{document}